\documentclass[a4paper,twocolumn,11pt,accepted=2025-12-01]{quantumarticle}
\pdfoutput=1

\usepackage[dvipsnames, table]{xcolor}
\usepackage[german,english]{babel}

\usepackage[utf8]{inputenc}

\usepackage[T1]{fontenc}
\usepackage{amsmath}
\usepackage[numbers,sort&compress]{natbib}
\usepackage{hyperref}
\usepackage{amsfonts}
\usepackage{amssymb}

\usepackage{tikz}
\usepackage{lipsum}
\usepackage{graphicx}
\usepackage{dcolumn}
\usepackage{bm}
\usepackage{physics}

\usepackage{subfigure}

\usepackage{tabularx}
\usepackage[normalem]{ulem}
\usepackage[linesnumbered, ruled]{algorithm2e}

\usepackage{colortbl}
\usepackage{diagbox}
\usepackage[version=4]{mhchem}

\usepackage[page]{appendix}

\usepackage{theorem}
\newtheorem{theorem}{Theorem}
\newtheorem{definition}{Definition}
\newtheorem{observation}{Observation}

\newtheorem{proposition}{Proposition}

\definecolor{myrefcolor}{rgb}{0.067,0.5,0.5}

\newenvironment{proof}{\noindent \textbf{{Proof~} }}{\hfill $\blacksquare$}

\newcommand{\sgn}{\operatorname{sgn}}

\begin{document}

\title{Exploiting many-body localization for scalable variational quantum simulation}

\author{Chenfeng Cao}
\orcid{0000-0001-5589-7503}
\email{chenfeng.cao@connect.ust.hk}
\affiliation{HK Institute of Quantum Science $\&$ Technology, The University of Hong Kong, Hong Kong, China}
\affiliation{Dahlem Center for Complex Quantum Systems, Freie Universität Berlin, 14195 Berlin, Germany}

\author{Yeqing Zhou}
\affiliation{Department of Physics, University of Wisconsin-Madison, Madison, WI 53706, USA}

\author{Swamit Tannu}
\affiliation{Department of Computer Science, University of Wisconsin-Madison, Madison, WI 53706, USA}

\author{Nic Shannon}
\affiliation{Theory of Quantum Matter Unit, Okinawa Institute of Science and Technology Graduate University, Onna-son, Okinawa 904-0412, Japan}

\author{Robert Joynt}
\affiliation{Department of Physics, University of Wisconsin-Madison, Madison, WI 53706, USA}

\affiliation{Theoretical Sciences Visiting Program (TSVP), Okinawa Institute of Science and Technology Graduate University, Onna, 904-0495, Japan}

\maketitle

\begin{abstract}
Variational quantum algorithms (VQAs) represent a promising pathway toward achieving practical quantum advantage on near-term hardware. Despite this promise, for generic, expressive ansätze, their scalability is critically hindered by barren plateaus—regimes of exponentially vanishing gradients. We demonstrate that initializing a hardware-efficient, Floquet-structured ansatz within the many-body localized (MBL) phase mitigates barren plateaus and enhances algorithmic trainability. Through analysis of the inverse participation ratio, entanglement entropy, and a novel low-weight stabilizer Rényi entropy, we characterize a distinct MBL--thermalization transition. Below a critical kick strength, the circuit avoids forming a unitary 2-design, exhibits robust area-law entanglement, and maintains non-vanishing gradients. Leveraging this MBL regime facilitates the efficient variational preparation of ground states for several model Hamiltonians with significantly reduced computational resources. Crucially, experiments on a 127-qubit superconducting processor provide evidence for the preservation of trainable gradients in the MBL phase for a kicked Heisenberg chain, validating our approach on contemporary noisy hardware. Our findings position MBL-based initialization as a viable strategy for developing scalable VQAs and motivate broader integration of localization into quantum algorithm design.
\end{abstract}

\section{Introduction}

Quantum computing promises to solve classically intractable problems in cryptography, optimization, and quantum chemistry. A key near-term approach is the variational quantum algorithm (VQA), a hybrid framework using parameterized quantum circuits and classical optimizers to mitigate hardware limits like coherence times and connectivity, tackling tasks like ground-state preparation via the Variational Quantum Eigensolver (VQE) for molecules~\cite{Peruzzo2014Quantum, Wecker2015Progress} or combinatorial optimization with the Quantum Approximate Optimization Algorithm (QAOA)~\cite{farhi2014quantum, Zhou2020Quantum}. Nevertheless, scaling VQAs faces major challenges, most notably barren plateaus—exponential gradient decay with system size~\cite{McClean2018Barren, Marrero2021Entanglement, Larocca2022diagnosingbarren, CerveroMartin2023barrenplateausin, larocca2024review}. These arise from various sources such as high expressivity yielding unitary 2-designs~\cite{McClean2018Barren} and volume-law entanglement~\cite{Marrero2021Entanglement}, unified in a Lie-algebraic framework~\cite{Ragone2024Lie, Fontana2024Characterizing}. Trainability is further challenged by spurious local minima~\cite{Anschuetz2022Quantum, Anschuetz2025Unified, Cao2025Unveiling} and noise effects~\cite{Mele2024Noise, Wang2021Noise, Quek2024Exponentially, Stilck2021Limitations, Cerezo2022Challenges}.

Entanglement's role is paradoxical: vital for quantum advantage, yet excessive amounts cause issues, such as inefficient ``CQ--universality'' (classical‑input/quantum‑output) in measurement-based quantum computing~\cite{Gross2009Most, Briegel2009Measurement} and barren plateaus in VQAs that hinder gradient-based and gradient-free optimizers through equivalent cost concentration~\cite{Arrasmith2021Effect, Arrasmith2022Equivalence}. Mitigation strategies target circuit and cost function structure constraints~\cite{Cerezo2021Cost, zhang2024absence, Wiersema2020Exploring, Pesah2021Absence, Zhao2021Analyzing, Park2024Hamiltonian, Wang2024Entanglement, Roeland2023Measurement}, altering the training protocol~\cite{Mele2022Avoiding, Sack2022Avoiding, Puig2025Variational}, and employing specialized initializations~\cite{Patti2021Entanglement, Wang2024Trainability, Park2024Hamiltonian}. However, it is worth noting that despite these efforts, the absence of barren plateaus alone does not guarantee a quantum advantage~\cite{Cerezo2023Does}. Recent analyses highlight that successful training may still require increasingly precise initializations~\cite{Mhiri2025Unifying} and that even trainable regions might be classically surrogatable with only polynomial overhead (after potentially identifying necessary basis with quantum computers)~\cite{Lerch2024Efficient}.

We address these issues using \textit{many-body localization} (MBL), a phenomenon in disordered quantum systems featuring a thermal-to-localized transition with increasing disorder~\cite{Anderson1958Absence, Basko2006Metal, Nandkishore2015Many, Abanin2019Colloquium, sierant2024manybody}. Thermal phases show volume-law entanglement entropy~\cite{Eisert2010Colloquium}, while MBL phases exhibit area-law scaling, preserving localization. MBL's local integrals of motion and decoherence resilience~\cite{Abanin2019Colloquium} have spurred quantum computing applications, including annealing acceleration~\cite{Cao2021Speedup}, VQE probing~\cite{Liu2023Probing}, analog expressibility~\cite{Tangpanitanon2020Expressibility}, Floquet Trotterization~\cite{Su2022Signatures, Zhu2021Probing}, and experimental integral extraction~\cite{shtanko2023uncovering}.

Recent works have also linked MBL to improved VQA trainability, in both digital hardware-efficient circuits~\cite{park2024hardwareefficient} and analog VQAs where phases of matter are closely tied to the flatness of the loss landscape~\cite{Srimahajariyapong2025Connecting}. In contrast, our work uses a distinct circuit architecture and trial-state construction, provides a $t$-design-based analytical explanation, and includes an experimental demonstration on a 127-qubit processor.

Inspired by Floquet MBL dynamics~\cite{Ponte2015Many-Body, Lazarides2015Fate, Abanin2016Theory}, we develop a Floquet-initialized, hardware-efficient, universal variational circuit exhibiting an MBL--thermalization transition~\cite{Ponte2015Many-Body, Lazarides2015Fate}. By initializing the circuit within the MBL phase—characterized by area-law entanglement and robust gradients—our protocol systematically avoids barren plateaus. The optimization begins from a low-complexity trial state and, while expressibility is unrestricted during training, our data suggest that for the Hamiltonians and system sizes studied, optimization trajectories remain in trainable regions, far from the 2-design limit that causes plateaus.

We validate our proposal with tools linking $t$-designs to MBL: a relation via inverse participation ratio (Theorem~\ref{theorem:IPR}) and a new low-weight stabilizer Rényi entropy $M_{t,k}$ correlating with higher-order ($2t$-) designs (Theorem~\ref{theorem:SE}). These results provide insight into circuit randomness, and under plausible assumptions,  provide evidence for barren-plateau absence in MBL (Observation~\ref{Claim:violatingWBP}, Theorem~\ref{theorem:grad_bound}). On the experimental side, we implement our protocol on the 127-qubit \texttt{ibm\_brisbane} processor, successfully demonstrating gradient restoration for a kicked Heisenberg model of up to 31 qubits.

The paper is organized as follows: Section~\ref{Sec:Setup} details Floquet-initialized circuits; Section~\ref{Sec:MBLTransition} examines the MBL--thermalization transition; Section~\ref{Sec:Barrenplateaus} discusses trainability and Aubry-André applications; Section~\ref{Sec:IBMExperiment} presents experimental results; Section~\ref{Sec:Conclusions} summarizes findings and future directions. In Appendices, we provide proofs of universality, theorems, and additional analyses.

\section{Methods}~\label{Sec:Setup}

Our approach employs an $n$-qubit hardware-efficient layered ansatz. As depicted in Fig.~\ref{fig:schematic}(a), each of the $D$ layers consists of a fixed sequence of blocks containing single-qubit and two-qubit Pauli rotations. This ansatz is universal, capable of approximating any unitary transformation with sufficient depth (see Appendix~\ref{obproof:universality} for a constructive proof).

The variational optimization begins from a low-complexity, classically obtained trial state 
$|\psi_t\rangle$ (e.g., a Hartree–Fock Slater determinant or a compact matrix-product state~\cite{hartree_1928_1, Verstraete2008Matrix, Schollwock2011The}). 
The construction of $|\psi_t\rangle$ aims to capture essential features of the target ground state's entanglement structure. 
Let $\rho_t := |\psi_t\rangle\langle\psi_t|$. 
We then apply a low-depth quantum channel $\mathcal E_t$ that maps this trial state to a simple reference state in the computational basis, 
$\rho_c := \mathcal E_t(\rho_t) = |\varphi_c\rangle\langle\varphi_c|$ (e.g., $|\uparrow\rangle^{\otimes n}$). 
The main variational circuit $\hat{U}(\boldsymbol{\vartheta})$ is then applied to $|\varphi_c\rangle$, and a decoding channel $\mathcal E_t^{-1}$ reverses the encoding defined by $\mathcal E_t$ on the relevant family of states explored. When viewed in the logical frame defined by $\mathcal E_t$, the evolution is effectively conjugated by $\mathcal E_t$ and $\mathcal E_t^{-1}$. This ensures that the evolution remains quasi-local with respect to the degrees of freedom already localized in $|\psi_t\rangle$, thereby preserving MBL-like properties at initialization. 
Ultimately, our method should be viewed as the second step in an optimization scheme in which the first step is the best possible classical computation.

\begin{figure}[t]
\centering
\includegraphics[width=8cm]{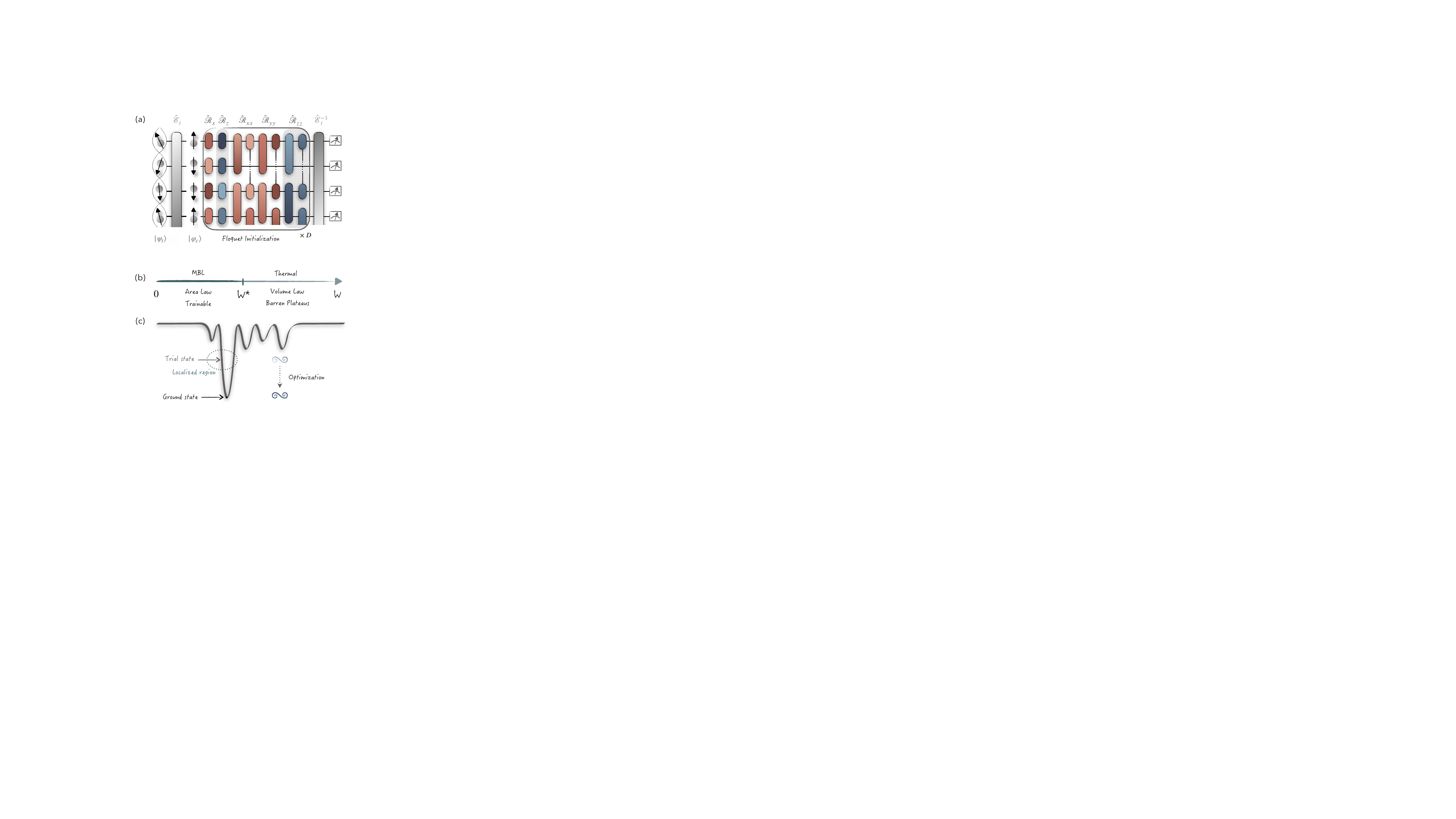}
\caption{Schematic representation of the Floquet-initialized variational quantum circuit, the localization-thermalization phase transition, and the optimization landscape. \textbf{(a)} Diagram of a variational quantum circuit. The illustration includes four qubits as a segment of a larger circuit, resulting in some 2-qubit gates acting on adjacent pairs $Q_i-Q_j$ appearing truncated. The circuit begins with a low-complexity trial state. Steady rotation angles within $\hat{\mathcal{R}}_z$ and $\hat{\mathcal{R}}_{zz}$ are initialized randomly within $[-\pi,\pi)$, while the kick rotation angles in $\hat{\mathcal{R}}_x$, $\hat{\mathcal{R}}_{xx}$, and $\hat{\mathcal{R}}_{yy}$ are initialized within $[-W,W]$. \textbf{(b)} Transition of the variational quantum circuit between the many-body localized (MBL) phase and the thermal phase as a function of the kick strength $W$. Below a critical threshold $W^\star$, the circuit remains in the MBL phase, characterized by output state entanglement entropy adhering to the area law and trainable circuit parameters. Above $W^\star$, the circuit enters the thermal phase, where the output state's entanglement entropy follows the volume law, rendering the circuit parameters untrainable. \textbf{(c)} Depiction of the optimization landscape for an MBL-initialized circuit. The initial output state is localized to a trial state, ideally positioned in the same ``valley'' as the ground state, facilitating effective optimization.}
\label{fig:schematic}
\end{figure}

The main variational circuit contains parameters that implement the Floquet initialization:
\begin{equation}\label{eq:1}
\hat{U}(\boldsymbol{\vartheta}) = \prod_{\ell=1}^D \bigl(\prod_{\substack{\alpha\in\\\{x,y,z\}}}\hat{\mathcal{R}}_{\alpha\alpha}(\boldsymbol{\vartheta}^{(\ell)}_{\alpha\alpha})\bigr)\bigl(\prod_{\alpha \in \{x,z\}}\hat{\mathcal{R}}_{\alpha}(\boldsymbol{\vartheta}^{(\ell)}_{\alpha})\bigr).
\end{equation}
Here, $D$ denotes the circuit depth, measured by the number of layers. For Pauli index $\alpha \in \{x, y, z\}$, we define:
\begin{equation}\label{eq:2}
\begin{aligned}
&\hat{\mathcal{R}}_{\alpha}(\boldsymbol{\vartheta}^{(\ell)}_{\alpha}) = \prod_{j=1}^n\exp\bigl(-i\frac{\vartheta^{(\ell j)}_{\alpha}}{2}\hat{\sigma}^\alpha_{j}\bigr),\\
&\hat{\mathcal{R}}_{\alpha\alpha}(\boldsymbol{\vartheta}^{(\ell)}_{\alpha\alpha}) = \prod_{\substack{(j, k) \in \mathrm{edges}(G)}}\exp\bigl(-i\frac{\vartheta^{(\ell jk)}_{\alpha\alpha}}{2}\hat{\sigma}^{\alpha}_{j}\hat{\sigma}^{\alpha}_{k}\bigr),
\end{aligned}
\end{equation}
where $\boldsymbol{\vartheta}_{\alpha}^{(\ell)}$ represents the rotation angles for single-qubit $\alpha$-Pauli rotations in the $\ell$-th layer, and $\boldsymbol{\vartheta}^{(\ell)}_{\alpha\alpha}$ denotes those for two-qubit $\alpha\alpha$-Pauli rotations in the same layer. The graph $G$ specifies the qubit connectivity for two-qubit gates.

The Floquet initialization begins by uniformly sampling first-layer \textit{steady parameters} $\vartheta^{(1j)}_{z}$ and $\vartheta^{(1jk)}_{zz}$ from $[-\pi,\pi)$, and \textit{kick parameters} $\vartheta^{(1j)}_{x}$, $\vartheta^{(1jk)}_{xx}$, $\vartheta^{(1jk)}_{yy}$ from $[-W, W]$, where $W$ is the kick strength. These are replicated across all layers $\ell=1,2,\dots,D$:
\begin{equation}\label{eq:3}
\boldsymbol{\vartheta}_\alpha^{(\ell)}      = \boldsymbol{\vartheta}_\alpha^{(1)},\quad
\boldsymbol{\vartheta}_{\alpha\alpha}^{(\ell)} = \boldsymbol{\vartheta}_{\alpha\alpha}^{(1)},
\quad \forall\, \ell \in \{1,\dots,D\}.
\end{equation}
yielding an initial parameter instance $\boldsymbol{\vartheta}_{\text{init}}$ drawn from the periodic ensemble $\Theta_\text{prd}(W)$.

This Floquet initialization corresponds to a digital simulation of a periodically driven many-body system. 
At initialization (when $\boldsymbol{\vartheta} \in \Theta_{\mathrm{prd}}(W)$), all layers share identical parameters, so the single-layer unitary $\hat{U}_1$ acts as the Floquet operator and the initialized circuit is $\hat{U}_1^{D}$. 
Since gates within each block commute [Eq.~\eqref{eq:2}], each block is the exact exponential of a (few-body) Hamiltonian term. 
Defining, for $\alpha\in\{x,z\}$,
\begin{equation}
\hat H_{\alpha} = \frac12 \sum_{j} \vartheta^{(1j)}_{\alpha}\hat\sigma^\alpha_{j},
\end{equation}
and for $\alpha\in\{x,y,z\}$,
\begin{equation}
\hat H_{\alpha\alpha} = \frac12 \sum_{\substack{(j, k) \in \mathrm{edges}(G)}} \vartheta^{(1jk)}_{\alpha\alpha}\hat\sigma^\alpha_{j}\hat\sigma^\alpha_{k},
\end{equation}
the Floquet operator (up to the chosen block ordering within a layer) can be written as
\begin{equation}
\hat{U}_1
=
e^{-i\hat H_{zz}}
e^{-i\hat H_{yy}}
e^{-i\hat H_{xx}}
e^{-i\hat H_{z}}
e^{-i\hat H_{x}}.
\end{equation}
One convenient representation is a piecewise-constant drive over one period $T$,
e.g. $\hat H(t)=\hat H_x$ for $t\in[0,1)$, $\hat H_z$ for $t\in[1,2)$, $\hat H_{xx}$ for $t\in[2,3)$,
$\hat H_{yy}$ for $t\in[3,4)$, and $\hat H_{zz}$ for $t\in[4,5)$ (and periodic), which yields $\hat U_1=\mathcal T e^{-i\int_0^T dt \hat H(t)}$, with $\mathcal T$ denoting time ordering. The steady angles are sampled uniformly from $[-\pi,\pi)$, which provides strong quenched disorder; Floquet-MBL is known to persist in such disordered interacting drives for sufficiently weak driving~\cite{Ponte2015Many-Body,Lazarides2015Fate,Abanin2016Theory,Abanin2019Colloquium}, and has been demonstrated numerically in related disordered Floquet spin-chain models~\cite{Zhang2016A}. 
In our parametrization, $W$ controls the typical magnitude of the kick generators: small $W$ yields localized (Floquet-MBL) dynamics dominated by the steady disorder, while large $W$ induces thermalization.

A subtlety specific to periodically driven systems is the distinction between genuine Floquet-MBL (no heating even at asymptotically long times) and Floquet prethermal regimes, where heating is parametrically slow and the dynamics is well captured by an effective quasi-conserved Hamiltonian over exponentially long times before eventual heating sets in~\cite{Mori2016Rigorous, Abanin2017Effective, Mori2018Thermalization}. In the near-term variational setting considered here, however, this distinction is largely immaterial: our circuits have finite depth and run on noisy hardware with finite coherence, so what matters is simply the existence of a long-lived non-ergodic window over the explored depths, during which entanglement remains low and the circuit stays far from unitary-design behavior. Accordingly, we use the term ``MBL phase'' operationally to denote this trainable non-ergodic regime.

This Floquet initialization is fundamentally distinct from close-to-identity schemes~\cite{Zhang2022Escaping, Wang2024Trainability, Park2024Hamiltonian}, where layer parameters are typically drawn independently near the identity. Here, strict layer replication ($\hat{U}_{\ell}=\hat{U}_1$ at initialization) is essential: it realizes stroboscopic evolution under a fixed disordered Floquet drive, as required in Floquet-MBL constructions.

For each layer unitary $\hat U_\ell$ we denote its induced channel by $\mathcal U_\ell(\rho):=\hat U_\ell\rho \hat U_\ell^\dagger$. It is convenient to describe the dynamics in the logical frame defined by the encoding channel $\mathcal E_t$, in which each physical layer $\mathcal U_\ell$ appears conjugated as $\mathcal E_t^{-1}\circ\mathcal U_\ell\circ\mathcal E_t$. 
In this frame, the initialized circuit acts as
\begin{equation}
\begin{aligned}
\rho(\boldsymbol{\vartheta}_{\mathrm{init}})
&=\bigl(\mathcal E_{t}^{-1} \circ \mathcal U_{D} \circ\mathcal E_{t}\bigr)
\cdots
\bigl(\mathcal E_{t}^{-1}\circ \mathcal U_{1}\circ \mathcal E_{t}\bigr)
(\rho_{t}) \\
&=\bigl(\mathcal E_{t}^{-1}\circ\mathcal U_{1}\circ \mathcal E_{t}\bigr)^{\circ D}
(\rho_{t}) \\ &=(\mathcal E_{t}^{-1}\circ \mathcal U^{\circ D}_{1}\circ \mathcal E_{t})
(\rho_{t})
\end{aligned}
\label{eq:effective-Floquet-layer}
\end{equation}
Each ``effective'' layer $(\mathcal E_{t}^{-1}\circ \mathcal U_{\ell}\circ \mathcal E_{t})$ acts in the logical basis where $|\psi_{t}\rangle$ is a simple product state; hence its generators are quasi-local in the $\mathcal E_t$-defined logical basis (i.e., with respect to the degrees of freedom captured by the trial state). In the sense of circuit expressibility defined in Refs.~\cite{Larocca2022diagnosingbarren, Holmes2022Expressibility}, Floquet replication greatly restricts the effective exploration of parameter space compared to fully random layer-wise parametrizations, leading to a slower increase in expressibility than in fully random circuits and keeping it, for the system sizes studied, below the unitary-2-design threshold.

During optimization, rotation angles in different layers are updated independently to minimize the output energy. With sufficient depth and suitable parameters, the variational quantum circuit can in principle prepare the exact ground state of the target Hamiltonian. While the output state is not confined to the vicinity of $|\psi_t\rangle$ during optimization—indicating release from strict MBL constraints—numerical evidence in Sec.~\ref{Sec:Barrenplateaus} suggests that MBL-based initialization keeps the optimization trajectories far from the Haar regime for the system sizes and models studied, and we did not observe barren-plateau behavior along these trajectories. A more systematic validation across broader settings and larger system sizes is left to future work. In principle, $\mathcal E_t$ could be a general (possibly non-unitary) encoding channel. In this work, for implementability we take $\mathcal E_t(\rho)=U_t\rho U_t^\dagger$ with a shallow unitary $U_t$, so that $\mathcal E_t^{-1}(\rho)=U_t^\dagger \rho U_t$.

\section{MBL--Thermalization Transition}\label{Sec:MBLTransition}
Many-body localization (MBL) arises in disordered quantum systems with interacting constituents, where disorder can trigger a transition from ergodic to non-ergodic behavior~\cite{Anderson1958Absence, Basko2006Metal, Nandkishore2015Many, Abanin2019Colloquium}. Consequently the system fails to thermalize and retains memory of its initial conditions. In this section, we analyze the MBL--thermalization phase transition in our setup using three diagnostics: inverse participation ratio (IPR), entanglement entropy, and a low‑weight stabilizer Rényi entropy (SRE).

We briefly recall the concept of quantum $t$-designs, critical for analyzing random quantum states and channels~\cite{Renes2004Symmetric, Gross2007Evenly, Dankert2009Exact}. An ensemble of parameterized unitary operations $\hat{U}(\boldsymbol{\vartheta}) $, with $\boldsymbol{\vartheta}$ uniformly sampled from the parameter space $\Theta$, forms a $t$-design if and only if it satisfies
\begin{equation}
    \int_{\Theta} (\hat{U}(\boldsymbol{\vartheta}) \otimes \hat{U}(\boldsymbol{\vartheta})^\dagger)^{\otimes t} d\boldsymbol{\vartheta} = \int_{\mathrm{U}(2^n)} (\hat{U} \otimes \hat{U}^\dagger)^{\otimes t} d\mu_{\mathrm H}(\hat{U}),
\end{equation}
where $\mathrm{U}(2^n)$ denotes the unitary group and
$d\mu_{\mathrm H}(\hat{U})$ is the Haar measure, i.e. the unique left- and right-invariant probability measure on $\mathrm{U}(2^n)$ that serves as the “uniform” distribution over all unitaries~\cite{Watrous2018Theory, Mele2024Introduction}.

Within the context of variational quantum optimization, if the circuit ensemble forms a unitary 2-design then cost gradients concentrate exponentially in the system size, independently of the cost locality~\cite{McClean2018Barren, Cerezo2021Cost}. While real-valued, time-reversal-symmetric circuits are expected to thermalize to the circular orthogonal ensemble (COE)~\cite{Collins2006Integration, Alessio2014Long-time}, our hardware-efficient ansatz employs non-commuting generators that span $\mathfrak{su}(2^n)$. This structure generally breaks time-reversal symmetry, meaning the unitary cannot be represented solely by a real matrix. Consequently the thermal phase is governed by the circular unitary ensemble (CUE).  A level-statistics check (Appendix~\ref{Appendix:CUE}) shows the disorder-averaged adjacent-gap ratio approaching the CUE benchmark $\langle r\rangle\simeq 0.603$ (as opposed to $\simeq 0.536$ for COE and $\simeq 0.386$ for Poisson), justifying the use of Haar–$\mathrm U(2^n)$ moments in our $t$-design benchmarks.

\vspace{1em}

\textit{Inverse Participation Ratio}-- The inverse participation ratio (IPR) is a valuable metric for characterizing the localization properties of quantum states. For an $n$-qubit state $|\psi\rangle$, the $t$-th-order IPR with respect to the basis states $\{|\beta_j\rangle\}$ is defined as
\begin{equation}
    \text{IPR}_t(|\psi \rangle) = \sum_{j} |\langle \beta_j |\psi \rangle |^{2t}.
\end{equation}
A higher $\text{IPR}_t$ value indicates a more localized state, while a lower value indicates greater delocalization. As $t$ increases, the IPR becomes increasingly sensitive to localization properties. In MBL studies, $\text{IPR}_2$ is commonly used to differentiate between localized and delocalized eigenstates~\cite{Luca2013Ergodicity, Luitz2015Universal}. Here, we select basis states $\{|\beta_j\rangle = \hat{U}_t^\dagger |j\rangle\}$ with $\{|j\rangle\}_{j=1,2,\dots,2^n}$ representing the computational basis. The trial state $|\psi_t\rangle$ belongs to $\{|\beta_j\rangle\}$. The $\text{IPR}_t$ measures how well the output state is localized to the trial state, noting that IPR's definition is basis-dependent.

There is a profound link between the inverse participation ratio and the $t$-design characteristic:
\begin{theorem}\label{theorem:IPR}
    Assume that $\{\hat{U}(\boldsymbol{\vartheta}) \}$, with $\boldsymbol{\vartheta}$ uniformly sampled from $\Theta$, forms a unitary $t$-design. For any arbitrary $n$-qubit input state $|\psi_{\text{in}}\rangle$, the expected value of the inverse participation ratio of the output state $|\psi(\boldsymbol{\vartheta})\rangle = \hat{U}(\boldsymbol{\vartheta}) |\psi_{\text{in}}\rangle$ with respect to some basis states $\{|\beta_j\rangle\}$ satisfies
    \begin{equation}
        \int_{\Theta} \text{IPR}_t(|\psi(\boldsymbol{\vartheta})\rangle)d\boldsymbol{\vartheta} = \frac{2^{n}!\, t!}{(t+2^n-1)!}.
    \end{equation}
For $t=2$, specifically:
\begin{equation}
    \int_{\Theta} \text{IPR}_2(|\psi(\boldsymbol{\vartheta})\rangle)d\boldsymbol{\vartheta} = \frac{2}{2^n+1}.
    \label{eq:IPR_2}
\end{equation}
\end{theorem}
A proof is available in Appendix~\ref{theoremproof:IPR}.

This result indicates that the inverse participation ratio not only quantifies localization but also provides a diagnostic for whether an ensemble of unitaries violates unitary $t$-design behavior. However, unlike the frame potential defined in Refs.~\cite{Renes2004Symmetric, sim2019expressibility}, $\text{IPR}_t$ is not suitable for quantifying the deviation of $\{|\psi(\boldsymbol{\vartheta})\rangle\}$ from $t$-designs with $t \geq 2$, as the basis states $\{|\beta_j\rangle\}$ do not form a state $2$-design. Instead, when $t\geq 2$, the dimension of the subspace spanned by the $t$-fold basis states is considerably smaller than the dimension of the $t$-fold symmetric subspace $\vee^t \mathbb{C}^{2^n}$~\cite{Harrow2013Church}:
\begin{equation}
\begin{aligned}
    \text{dim}\big(\text{Span}\{|j\rangle^{\otimes t}\}\big) &= 2^n \\&\ll \dim \bigl(\vee^{t} \mathbb C^{2^{n}}\bigr) \\&=\binom{t+2^n-1}{t}.
\end{aligned}
\end{equation}
Despite these dimensional differences, we confirm that in the thermal phase, the output state ensemble qualifies as a state 2-design. This confirmation comes from our direct calculations of frame potentials, detailed in Appendix~\ref{Sec: Frame Potential}. Additionally, we explore the expressibility of the Floquet initialization for both the MBL and thermal phases in the same appendix, providing insights into the dynamic behaviors of these phases.

\vspace{1em}

\textit{Entanglement Entropy}-- A pivotal characteristic of MBL is the behavior of the von Neumann entanglement entropy ($S_L$) of a subsystem $L$:
\begin{equation}
S_L = -\operatorname{Tr}(\rho_L \log \rho_L)
\end{equation}
where $\rho_L$ is the reduced density matrix for subsystem $L$. In thermalized systems, $S_L$ adheres to a volume law scaling, proportional to subsystem size: $S_L \sim |L|^d$, where $d$ denotes the spatial dimension of the system. In contrast, MBL systems exhibit an area law scaling, in which $S_L$ increases proportionally to the boundary of the subsystem $L$: $S_L \sim |L|^{d-1}$~\cite{Eisert2010Colloquium}. For one-dimensional systems, where $d=1$, $S_L$ essentially remains constant. The transition from volume law to area law scaling is a defining characteristic of MBL systems and has been extensively studied through numerical simulations~\cite{Pal2010Many, Luitz2015Universal}.

\vspace{1em}

\textit{Low-Weight Stabilizer Rényi Entropy--} The final metric we employ is a restricted version of the stabilizer Rényi entropy, as defined in Ref.~\cite{Leone2022Stabilizer}, which aims to quantify the magic resource (non-stabilizerness) of a quantum state. This metric provides valuable insights into quantum many-body systems, enhancing our understanding of complex patterns in ground-state wave functions~\cite{oliviero2022magic} and the dynamics of quantum quenches~\cite{rattacaso2023stabilizer}. Unlike the original definition, which encompasses all Pauli terms and is thus challenging to estimate in practical scenarios, our approach focuses solely on low-weight Pauli terms. We define the low-weight stabilizer Rényi entropy as follows:
\begin{definition}[low-weight SRE]
    Given an $n$-qubit state $|\psi\rangle$, we define its low-weight stabilizer Rényi entropy of order $t$ and locality $k$ as 
\begin{equation}\label{def:lowweight-sre}
    M_{t,k}(|\psi\rangle)=(1-t)^{-1} \log \sum_{\hat{P} \in \mathcal{P}_{n,k}} \frac{|\langle\psi|\hat{P}| \psi\rangle|^{2t}}{\operatorname{card}(\mathcal{P}_{n,k})},
\end{equation}
where $\mathcal{P}_{n,k}$ is the set of $n$-qubit Pauli strings (including the identity) with phase $+1$ and Pauli weight \(w(\hat P) \leq k\). 
\end{definition}
This metric characterizes the extent to which a quantum state $|\psi\rangle$ can be stabilized by some low-weight Pauli tensor products, a concept particularly relevant in quantum error correction contexts~\cite{Breuckmann2021Quantum}.

In the Floquet-MBL regime ($W<W^\star$), the dynamics admit an extensive set of quasi-local integrals of motion (the $\ell$-bits), $\{\hat\tau^z_j\}_{j=1}^n$. Deep in the localized phase, the memory of the initial computational basis state is preserved, implying that each $\hat\tau^z_j$ maintains an $\mathcal O(1)$ overlap with the local Pauli operator $\hat\sigma^z_j$. Consequently, there exist constants $\beta(W)>0$ and a subset $\mathcal{Z}_{n,k}\subseteq\mathcal{P}_{n,k}$ consisting of $Z$-strings with $|\mathcal{Z}_{n,k}| = \Theta(n^k)$ such that for all $\hat P_S\in\mathcal{Z}_{n,k}$ one has $|\langle\psi|\hat P_S|\psi\rangle|\ge \beta(W)$. Since the total cardinality of the Pauli set scales as $\lvert\mathcal P_{n,k}\rvert=\sum_{w=0}^k 3^w\binom{n}{w}=\Theta(3^k n^k)$, these structured correlations constitute a constant fraction $\Theta(3^{-k})$ of $\mathcal P_{n,k}$. Therefore, the average overlap scales as
\begin{equation}
\begin{aligned}
\frac{1}{\operatorname{card}(\mathcal{P}_{n,k})} \sum_{\hat{P} \in \mathcal{P}_{n,k}}
|\langle\psi|\hat{P}| \psi\rangle|^{2t}
&\ge \Theta(3^{-k})\,\beta(W)^{2t}
\\&= \Theta(1),
\end{aligned}
\label{eq:lbit-bound}
\end{equation}
which implies that $M_{t,k}$ remains $\mathcal O(1)$ for all $n$ (at fixed $t,k$) throughout the MBL phase ($W<W^\star$).

By contrast, in the thermal phase ($W>W^{\star}$), the eigenstate thermalization hypothesis suggests $|\langle\psi|\hat P|\psi\rangle|=\mathcal O(2^{-n/2})$ for any fixed-weight $\hat P$. Because $M_{t,k}$ is a moment-generating functional of the overlaps in Eq.~\eqref{def:lowweight-sre}, the bound in Eq.~\eqref{eq:lbit-bound} keeps $M_{t,k}$ of order unity in the MBL phase, whereas in the thermal phase the overlaps become exponentially small and $M_{t,k}$ approaches its Haar benchmark. This suggests that the difference between $M_{t,k}$ and its Haar benchmark serves as a potential order parameter for the MBL--thermalization transition.

Beyond its physical relevance, $M_{t,k}$ also quantifies how close the output-state ensemble is to forming higher-order designs, as it depends on the $(2t)$-th moments of the distribution:
\begin{theorem}\label{theorem:SE}
    Assume that $\{\hat{U}(\boldsymbol{\vartheta})\}$, with $\boldsymbol{\vartheta}$ uniformly sampled from $\Theta$, forms a unitary $2t$-design. Then any arbitrary $n$-qubit input state $|\psi_{\text{in}}\rangle$ transformed by these operations satisfies:
    \begin{equation}
\begin{aligned}
    &\int_{\Theta}  
  \frac{1}{\operatorname{card}(\mathcal P_{n,k})}
  \sum_{\hat P\in\mathcal P_{n,k}}
  \langle\psi(\boldsymbol \vartheta)|\hat P|\psi(\boldsymbol \vartheta)\rangle^{2t} d\boldsymbol \vartheta\\ = &\frac{\binom{2^{n} +2t - 1}{2t} - \binom{2^{n-1} +t - 1}{t}}{\operatorname{card}(\mathcal{P}_{n,k})\binom{2^{n} +2t - 1}{2t}}+\frac{\binom{2^{n-1} +t - 1}{t}}{\binom{2^{n} +2t - 1}{2t}}.
\end{aligned}
\end{equation}
Let this value be denoted by $T_{t,k}$. By Jensen's inequality, the expected low-weight stabilizer Rényi entropy of order $t$ and locality $k$ of $|\psi(\boldsymbol{\vartheta})\rangle$ satisfies:
\begin{equation}\label{eq:M_Haar}
\begin{aligned}
    \int_{\Theta} M_{t,k}(|\psi(\boldsymbol{\vartheta})\rangle) d\boldsymbol{\vartheta} \geq \frac{\log(T_{t,k})}{1-t}.
\end{aligned}
\end{equation}
\end{theorem}
Thus, computing $M_{t,k}$ provides a convenient witness against $2t$-design behavior from low-weight moments. We denote the right-hand side of Inequality~\eqref{eq:M_Haar} as $\widetilde {M}_{t,k}^{(\mathrm{Haar})}$. Specifically, $M_{2,2}$ probes fourth moments (hence 4-design signatures).
A significant deviation from the Haar benchmark indicates that the output ensemble is far from Haar-random at the level of fourth moments; in our setting this correlates with, and provides evidence for, staying away from the unitary-2-design regime relevant to barren plateaus. The proof for Theorem~\ref{theorem:SE} is provided in Appendix~\ref{theoremproof:SE}.

\vspace{1em}

\textit{Characterizing MBL and Thermalization--} Utilizing the inverse participation ratio, entanglement entropy, and the low-weight stabilizer Rényi entropy as metrics, we investigate the MBL--thermalization phase transition in a layered ``square circuit" whose layer depth equals the system size, \(D=n\). With \(n\) qubits, each layer contains \(\Theta(n)\) one- and two-qubit rotations, so the total number of elementary gates is \(\Theta(nD)=\Theta(n^{2})\). Such circuits are known to exhibit barren plateaus when all parameters are chosen independently at random~\cite{Cerezo2021Cost}.

Unless stated otherwise, the trial state is taken to be a product state
\begin{equation}
|\psi_t\rangle
   = \bigotimes_{j=1}^{n} U_j^{(\mathrm{Haar})}   |\uparrow \rangle_j,
\end{equation}
where each $U_j^{(\mathrm{Haar})}$ represents a single-qubit Haar-random unitary operation applied to qubit $Q_j$. A shallow unitary $\hat{U}_t$ then maps \(|\psi_t\rangle\) to the computational basis state of highest overlap.

We study two graphs:
\begin{itemize}
  \item A 1D ring, in which each qubit is coupled only to its nearest neighbors.
  \item A \emph{circulant graph} \(\mathrm{Ci}_n(1,2)\), in which each qubit \(Q_j\) is coupled to qubits \(Q_{j \pm 1}\) and \(Q_{j \pm 2}\) (mod \(n\)), i.e.\ edges of length 1 \emph{and} 2 around the ring.
\end{itemize}

\begin{figure}[t]
	\centering
	\includegraphics[width=8.25cm]{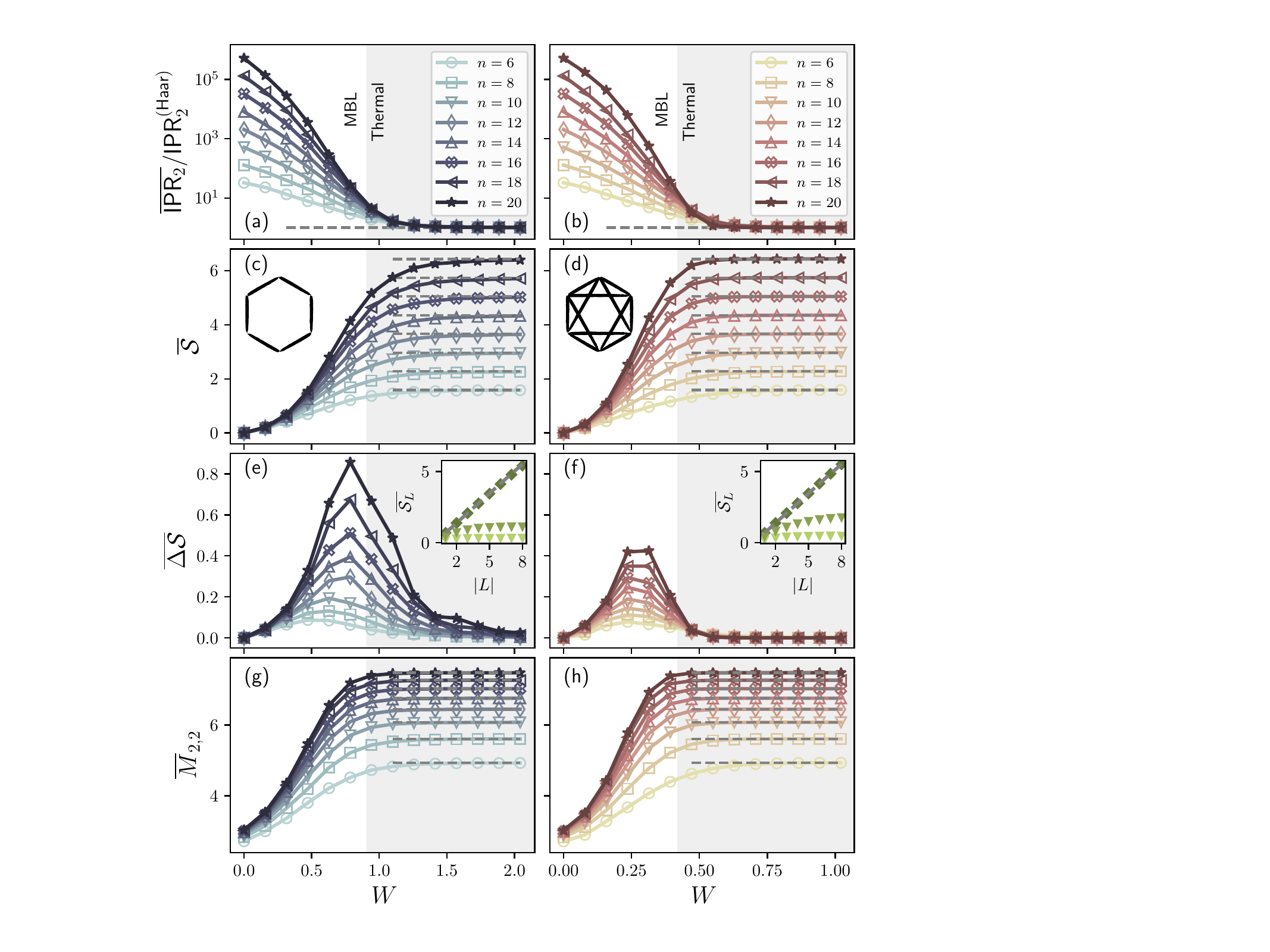}
	\caption{Characterization of the MBL--thermalization transition. Panels (a, c, e, g) represent results for the 1D ring topology, and panels (b, d, f, h) represent results for the $\mathrm{Ci}_{n}(1,2)$ topology. 
\textbf{(a, b)} Ratios of the inverse participation ratio, $\text{IPR}_2$ [Eq.~\eqref{eq:IPR_2}], to the Haar-random inverse participation ratio, $\text{IPR}_2^{(\mathrm{Haar})}$ [Eq.~\eqref{eq:IPR_Haar}], as a function of kick strength $W$. 
 The grey dashed lines correspond to $1$. 
\textbf{(c, d)} Half-chain entanglement entropy versus kick strength $W$. 
 The grey dashed lines correspond to $S^{(\mathrm{Page})}$ [Eq.~\eqref{eq:S^Page}]. 
 Insets: Demonstration of the topologies of a 6-qubit ring and a $\mathrm{Ci}_{6}(1,2)$ graph. 
\textbf{(e, f)} Entanglement entropy fluctuations versus kick strength $W$. 
 Insets: Entanglement entropy of a local region $L$ versus the size of $L$ for both MBL ($W = 1/5, 2/5$ for the 1D ring topology and $W = 1/10, 1/5$ for the $\mathrm{Ci}_{n}(1,2)$ topology) and thermal ($W = 7/5$ for 1D ring and $W = 7/10$ for $\mathrm{Ci}_{n}(1,2)$) phases in a 20-spin chain. 
 Darker green markers indicate circuits with higher $W$. 
 The grey dashed lines correspond to $S^{(\mathrm{Page})}$. 
\textbf{(g, h)} The low-weight stabilizer Rényi entropy of order $2$ and locality $2$ [Eq.~\eqref{def:lowweight-sre}] versus kick strength $W$. 
 The grey dashed lines correspond to the lower bound $\widetilde {M}_{2,2}^{(\mathrm{Haar})}$ [Eq.~\eqref{eq:lowerbound_SE}].}
	\label{fig:MBL-delocalization transition}
\end{figure}

The results for average $\text{IPR}_2$, half-chain entanglement entropy and its fluctuations, and $M_{2,2}$ are depicted in Fig.~\ref{fig:MBL-delocalization transition} with each point averaged over at least 200 samples.

To systematically estimate the critical kick strength \(W^\star\) we use an operational, finite-size criterion.  
For each system size \(n\) we scan \(W\) on a fine grid, record three diagnostics  
(i) \(\log_{10}\overline{\mathrm{IPR}_2}\),  
(ii) half-chain entropy \(\overline{S}(n/2,n)\),  
(iii) low-weight stabilizer Rényi entropy \(\overline{M}_{2,2}\).  
Spline fits to these observables yield the transition points as the positions of maximal susceptibility  
\begin{equation}
  W^\star_{\mathcal X}(n)=
  \arg\max_W \bigl|\partial_W^{2}\mathcal X(W)\bigr|.
\end{equation}
for \(\mathcal X\in\{\log\mathrm{IPR}_2,S,M_{2,2}\}\). We quote the crossover as the average over the three observables and all \(n\) considered,
\begin{equation}
  W^\star=\frac{1}{3}\bigl(
    \overline{W}^\star_{\log\mathrm{IPR}_2}
   +\overline{W}^\star_{S}
   +\overline{W}^\star_{M_{2,2}}
  \bigr),
\end{equation}
giving \(W^\star\approx 0.90\) for the 1D ring and \(W^\star\approx0.42\) for \(\mathrm{Ci}_{n}(1,2)\); the spread among the three estimates is around \(0.1\).

For both graph types, when the kick strength is below a certain threshold $W^\star$, the output state remains localized to the trial state:
\begin{equation}
    \overline{\text{IPR}}_2 \gg \text{IPR}_2^{(\mathrm{Haar})} =\frac{2}{2^n+1};
    \label{eq:IPR_Haar}
\end{equation}
the average half-chain entanglement entropy, $\overline{S}(\frac{n}{2},n)$, is significantly lower than the Page value for a random pure state, $S^{(\mathrm{Page})}(\frac{n}{2},n)$~\cite{Page1993Average, Foong1994Proof}:
\begin{equation}\label{eq:S^Page}
\begin{aligned}
    \overline{S} \left(\frac{n}{2},n\right)
      &\ll S^{(\mathrm{Page})} \left(\frac{n}{2},n\right)
      \\&= \sum_{j=2^{n / 2}+1}^{2^n} \frac{1}{j}-\frac{2^{n / 2}-1}{2 \cdot 2^{n / 2}};
\end{aligned}
\end{equation}
the entanglement entropy fluctuation, quantified by the sample variance 
$\overline{\Delta \mathcal{S}} := \mathrm{Var}[S(n/2,n)]$, is substantial; and the output state can be stabilized by some low-weight Pauli stabilizers to a certain extent
\begin{equation}
\begin{aligned}
    \overline{M}_{2,2} \ll M_{2,2}^{(\mathrm{Haar})} = &\log (\operatorname{card}(\mathcal{P}_{n,2})) - \mathcal{O}(2^{-n}).
\end{aligned}
\end{equation}
A convenient analytic lower bound on
$M_{2,2}^{(\mathrm{Haar})}$ obtained from the exact
moment formula with $2t=4$ is
\begin{equation}
\begin{aligned}
\label{eq:lowerbound_SE}
  \widetilde M_{2,2}^{(\mathrm{Haar})}
  = &-\log \biggl(
      \frac{1}{\operatorname{card}\mathcal P_{n,2}}
      \\&+\left(1-\frac{1}{\operatorname{card}\mathcal P_{n,2}}\right)
        \frac{\binom{2^{n-1}+1}{2}}
             {\binom{2^{n}+3}{4}}
    \biggr).
\end{aligned}
\end{equation}
When $W>W^\star$, $\text{IPR}_2$ approaches $2/(2^{n}+1)$,
$\overline{S}$ approaches $S^{(\mathrm{Page})}$,
$\overline{\Delta\mathcal S}$ collapses to zero, and
$\overline{M}_{2,2}$ gets close to 
\begin{equation}
    \log (\operatorname{card}(\mathcal{P}_{n,2})) = \log(\frac{9n^2}{2}-\frac{3n}{2}+1).
\end{equation}
Insets for a 20-qubit system illustrate that the entanglement entropy in the MBL phase (with $W = 1/5, 2/5$ for the 1D ring and $W = 1/10, 2/10$ for $\mathrm{Ci}_{n}(1,2)$) follows the area law, while entanglement entropy in the thermal phase adheres to the volume law. It is conjectured (and has been proven for 1D systems) that the entanglement of the ground states of gapped spin systems obeys the area law~\cite{Eisert2010Colloquium, Hastings2007An}. When preparing these states with variational quantum optimization, it is reasonable to initialize the variational quantum circuits in the MBL phase to maintain consistency in entanglement entropy and, preferably, entanglement structure. The optimal choice of the kick strength $W$ depends on the hardware topology, circuit depth, and the entanglement properties of the ground state of the target Hamiltonian. Though these numerical results primarily pertain to square circuits, local information remains well-preserved even in deeper circuits. Detailed illustrations are available in Appendix~\ref{Sec: Power Spectrum}.

\section{Variational Optimization}\label{Sec:Barrenplateaus}

This section examines the performance of the VQE in preparing the ground state of the Aubry-André model across different phases~\cite{Aubry1980Analyticity}. The Aubry–André model, extensively studied in condensed matter physics, characterizes a one-dimensional system of spinless fermions with nearest-neighbor hopping and a quasiperiodic on-site potential. The system's phase transitions are governed by the ratio of the quasiperiodic potential strength \(V\) to the hopping amplitude \(J\). Its Hamiltonian is:
\begin{equation}
\begin{aligned}
    \hat{\mathcal{H}}^{(\text{fm})}_{\text{AA}} = &-J \sum_{j} \bigl(\hat{c}_j^\dagger \hat{c}_{j+1} + \hat{c}_{j+1}^\dagger\hat{c}_j \bigr) + \Gamma \sum_{j}\hat{n}_j \hat{n}_{j+1} \\&+ V \sum_{j} \cos(2 \pi \alpha j + \phi) \hat{n}_j,
\end{aligned}
\end{equation}
where \( \hat{c}_j^\dagger \), \( \hat{n}_j \), and \( \phi \) denote the usual fermionic creation, number operators, and phase offset, respectively.

In the non-interacting limit ($\Gamma=0$), the model undergoes a transition from Anderson localization to an extended phase at a critical point determined by the ratio \(|V|/|J|\), with localization occurring when $|V| > 2|J|$ and delocalization when $|V| < 2|J|$. Upon introducing nearest-neighbor repulsive interactions ($\Gamma > 0$), the system exhibits transitions between MBL and ergodic phases~\cite{Iyer2013Many}. We emphasize that ``MBL initialization" and ``thermal initialization" refer exclusively to initial parameter ensembles of the variational circuit. Our protocol employs distinct trial states for different phases of the model, yet these initializations neither constrain the many-body phase of the target Hamiltonian nor restrict the optimization trajectory, as circuit parameters evolve freely after the initial optimization step.

Utilizing the Jordan-Wigner transformation~\cite{Jordan1928}, the qubit form of the Aubry-André model is derived as:
\begin{equation}
\label{Eq:AA Hamiltonian - qubit}
\begin{aligned}
\hat{\mathcal H}_{\text{AA}}
   &= -\frac{J}{2}\sum_{j}
        \bigl(
          \hat{\sigma}^{x}_{j}\hat{\sigma}^{x}_{j+1}
          +\hat{\sigma}^{y}_{j}\hat{\sigma}^{y}_{j+1}
        \bigr) +\frac{\Gamma}{4}\sum_{j}
        \hat{\sigma}^{z}_{j}\hat{\sigma}^{z}_{j+1} \\[2pt]
   &\quad -\frac{V}{2}\sum_{j}
        \bigl[
          \cos \bigl(2\pi\alpha j+\phi\bigr)
          +\Gamma/V
        \bigr]\hat{\sigma}^{z}_{j} .
\end{aligned}
\end{equation}

Before discussing variational optimization, we recall the notion of a \emph{weak barren plateau} introduced in Ref.\cite{Sack2022Avoiding}, the avoidance of which is sufficient to prevent barren plateaus. We also explain why bounding the von Neumann entropy is sufficient to exclude their occurrence.

An ensemble of states ${|\psi(\boldsymbol{\vartheta})\rangle}$ is said to exhibit a weak barren plateau if for every connected region $L$ of at most $k$ qubits the second Rényi entropy \(S_L^{(2)}:=-\log \left(\operatorname{Tr} (\rho_L^2)\right)\) satisfies
\begin{equation}
S^{(2)}_L \ge (1-\delta) S^{(\mathrm{Page})}(k,n),\quad 0<\delta<1,
\label{eq:WBP}\end{equation}
where $S^{(\mathrm{Page})}(k,n)$ is the Page entropy of a $k$-qubit subsystem of an $n$-qubit Haar-random pure state. Condition~\eqref{eq:WBP} implies that the variance of any $k$-local cost function scales as $\mathcal{O}(2^{-n})$, resulting in gradients that concentrate around zero~\cite{Sack2022Avoiding}.

Since Rényi entropies decrease monotonically with their order, we always have $S^{(2)}_L \le S_L$. Therefore, an upper bound of the form
\begin{equation}
S_L < \xi S^{(\mathrm{Page})}(k,n),\quad \xi < 1-\delta,
\label{eq:VNbound}
\end{equation}
precludes condition~\eqref{eq:WBP}, thus sufficing to rule out the occurrence of randomness-induced barren plateaus and ensuring gradients do not vanish exponentially~\cite{Sack2022Avoiding}.

\begin{observation}\label{Claim:violatingWBP}
For every $n$--qubit Hamiltonian $\hat{\mathcal{H}}$, every connected $k$--qubit region $L$, and every $\xi\in(0,1)$, there exists a kick strength $W\in[0,W^{\star})$ inside the MBL phase such that the Floquet--initialized parameter set $\Theta_{\mathrm{prd}}(W)$ fulfils
\begin{equation}
\mathbb{E}_{\boldsymbol{\vartheta} \in \Theta_{\mathrm{prd}}(W)}S_L  < \xi S^{(\mathrm{Page})}(k,n).
\end{equation}
Choosing $\xi<1-\delta$ immediately violates Eq.~\eqref{eq:WBP} at that $W$.
\end{observation}

Practically, this implies that initializing optimization within the MBL phase using sufficiently small kick strengths ensures gradients remain sufficiently large, effectively mitigating barren plateaus during early optimization stages.

\vspace{1em}

\textit{Barren Plateaus--} We numerically substantiate observation~\ref{Claim:violatingWBP} by computing the initial gradients before optimization, $\nabla \langle\hat{\mathcal{H}}\rangle_{\boldsymbol{\vartheta}}$, for the Floquet initialization with varying $W$. The trial state is drawn from an ensemble of products of single-qubit Haar-random states, and $\hat{U}_t$ transforms the trial state into a computational basis state $|\varphi_c\rangle$. The energy associated with the initial parameters $\boldsymbol{\vartheta}$ is
\begin{equation}
    \langle\hat{\mathcal{H}}\rangle_{\boldsymbol{\vartheta}} =
\langle \varphi_c| \hat{U}^\dagger(\boldsymbol{\vartheta})\hat{U}_t\hat{\mathcal{H}}_{\text{AA}}
\hat{U}_t^\dagger\hat{U}(\boldsymbol{\vartheta})|\varphi_c\rangle.
\end{equation}

For simplicity, we consider the critical point of the non-interacting Aubry-André model, setting $\alpha = (\sqrt{5}-1)/2$, $\phi = 0$, $J = 1$, $\Gamma = 0$, and $V = 2$. Fig.~\ref{fig:BP transition} plots the $\ell_\infty$-norm of the gradient as a function of $W$ for various system sizes, averaged over 1000 samples per point. In the MBL phase, the gradients remain stable as $n$ increases (though they decay with $W$), while in the thermal phase, the gradients decay exponentially with $n$, indicating the onset of barren plateaus. This gradient norm positively correlates with the deviation of $M_{t,k}$ from its Haar benchmark. Note that, although scale-independent points may appear within the MBL phase, they do not mark the precise transition to thermalization. Notably, most earlier studies characterize barren plateaus through the variance of the (zero-mean) gradient rather than its maximum component.
In Appendix~\ref{Appendix:inftynorm_variance} we prove Proposition~\ref{prop:inftynorm_variance}, which states that—under the usual zero-mean assumption—a non-negligible probability of observing a large $\ell_\infty$-gradient necessarily implies an inverse-polynomial lower bound on that variance.

\begin{figure}[t]
	\centering
	\includegraphics[width=8.1cm]{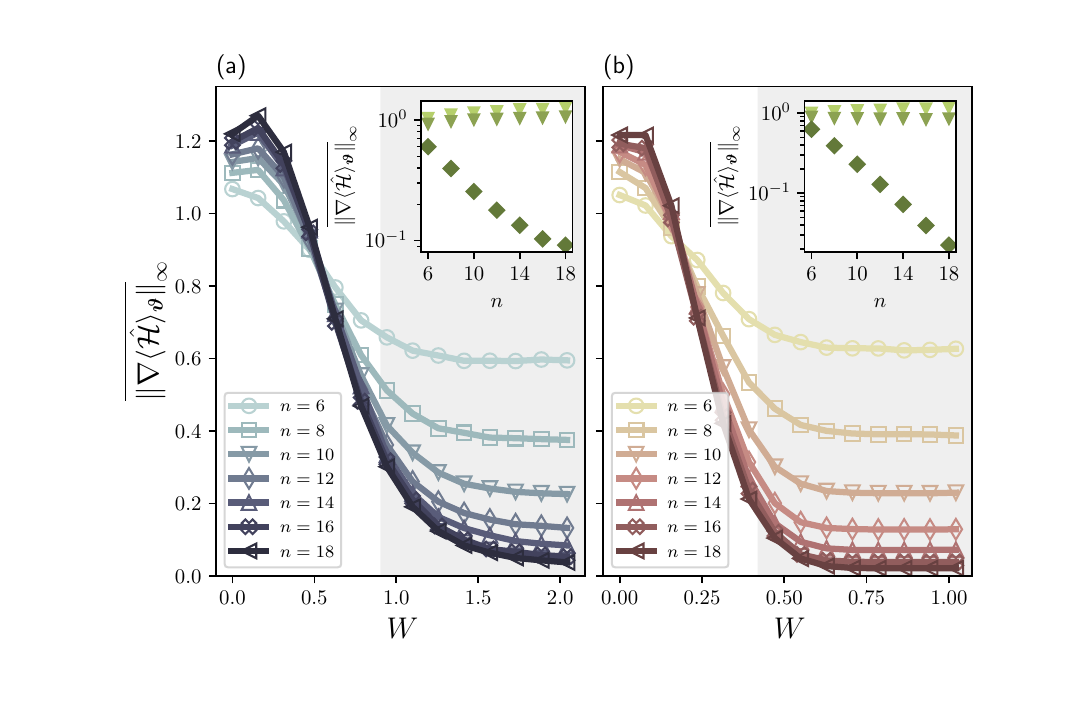}
	\caption{Visualization of the absence of barren plateaus in the MBL phase and their presence in the thermal phase after sampling the initial circuit parameters, analyzed using the target Hamiltonian $\hat{\mathcal{H}}_{\text{AA}}$ (defined in Eq.~\eqref{Eq:AA Hamiltonian - qubit}). \textbf{(a)} $\ell_\infty$-norm of the energy gradient plotted as a function of kick strength $W$ for a 1D ring topology. \textbf{(b)} $\ell_\infty$-norm of the energy gradient as a function of $W$ for a $\mathrm{Ci}_{n}(1,2)$ lattice topology. Insets illustrate the gradient scaling at various $W$ values, emphasizing the contrast between the MBL phase (e.g., $W = 1/5, 2/5$ for the 1D ring and $W = 1/10, 1/5$ for the $\mathrm{Ci}_{n}(1,2)$ topology) and the thermal phase (e.g., $W = 7/5$ for the 1D ring and $W = 7/10$ for the $\mathrm{Ci}_{n}(1,2)$ topology). Darker green markers in the insets denote circuits with higher $W$.}
	\label{fig:BP transition}
\end{figure}

Motivated by the information-preserving nature of the MBL phase, we define a uniform kick scaling operation and establish a theoretical lower bound on the gradient along a linear path connecting MBL- and thermally-initialized parameters.

\begin{definition}[Uniform kick scaling]
\label{def:uniform-kick-scaling}
Let $\boldsymbol\vartheta = (\boldsymbol\vartheta^{(1)}, \dots, \boldsymbol\vartheta^{(D)}) \in \mathbb{R}^{D( m_s +  m_k)}$ be a Floquet-initialized parameter point, i.e., $\boldsymbol\vartheta^{(\ell)} = \boldsymbol\vartheta^{(1)}$ for all $\ell$. \(m_s\) and \(m_k\) are the number of steady and kick parameters per layer, respectively. Given a scale factor $\lambda > 1$, we define the \emph{uniform kick scaling} operation as the map
\begin{equation}
(\vartheta_{j}^{(\ell)})_{\ell=1}^{D}\longmapsto
  (\lambda\vartheta_{j}^{(\ell)})_{\ell=1}^{D},
\end{equation}
for all kick parameters  \(j\)
while leaving the steady parameters unchanged.
\end{definition}

\begin{theorem}
\label{theorem:grad_bound}
Let \( U(\boldsymbol{\vartheta}) \) be a depth-\(D\) layered circuit, Floquet-initialized at \( \boldsymbol{\vartheta}_{\mathrm{MBL}} \in \Theta_{\mathrm{prd}}(W) \) with \( W < W^\star \), and let \( \hat{\mathcal{H}} \) be a Hermitian observable with \( \|\hat{\mathcal{H}}\| \le 1 \). Let \( \boldsymbol{\vartheta}_{\mathrm{th}} \) be obtained by scaling all kick parameters of \( \boldsymbol{\vartheta}_{\mathrm{MBL}} \) by a factor \( \lambda > 1 \), and define the linear interpolation
\begin{equation}
\gamma(s) := (1 - s)\boldsymbol{\vartheta}_{\mathrm{MBL}} + s  \boldsymbol{\vartheta}_{\mathrm{th}}, \quad s \in [0,1].
\end{equation}
If \( \langle \hat{\mathcal{H}} \rangle_{\boldsymbol\vartheta_{\mathrm{th}}} \) is exponentially small in \( n \), and \(D, m_k = \operatorname{poly}(n)\), where \( m_k \) is the number of kick parameters per layer, then
\begin{equation}
\label{eq:grad-lower-bound}
\mathbb{E}_{s \in [0,1]} \left\| \nabla_{\boldsymbol\vartheta} \langle\hat{\mathcal{H}}\rangle \big(\gamma(s)\big) \right\|_\infty
 \gtrsim 
\frac{ \lvert \langle\hat{\mathcal{H}}\rangle_{\boldsymbol\vartheta_{\mathrm{MBL}}} \rvert }
{ D m_k (\lambda - 1) W},
\end{equation}
which remains inverse-polynomial in \( n \) whenever \( \langle\hat{\mathcal{H}}\rangle_{\boldsymbol\vartheta_{\mathrm{MBL}}} = \Theta(1) \).
\end{theorem}
The proof for Theorem~\ref{theorem:grad_bound} is provided in Appendix~\ref{appendix:grad_bound}. This bound is meaningful in the MBL phase because its non-ergodic nature ensures the expectation value \(\langle \hat{\mathcal{H}} \rangle_{\boldsymbol\vartheta_{\mathrm{MBL}}}\) remains of order unity. In contrast, this expectation would be exponentially suppressed in the thermal phase, rendering the bound trivial.

\textit{Optimization--} We proceed with the VQE optimization process by fixing the Hamiltonian parameters to $\alpha = (\sqrt{5}-1)/2$, $\phi = 0$, and $J=1$. We employ the variational quantum eigensolver to prepare ground states of the one-dimensional $\hat{\mathcal{H}}_{\text{AA}}$ in different phases: non-interacting Anderson localized phase ($\Gamma=0, V = 4$), extended phase ($\Gamma=0, V = 1$), MBL phase ($\Gamma=3, V = 6$), and ergodic phase ($\Gamma=1, V = 2$). For the variational ansatz, we use a 1D ring connectivity square circuit with a depth $D=n$. We investigate three distinct initialization strategies as follows:
\begin{itemize}
\item MBL initialization: Floquet initialization with steady parameters uniformly sampled from $[-\pi,\pi)$ and kick parameters from $[-2/5, 2/5]$.
\item Thermal initialization: Floquet initialization with steady parameters uniformly sampled from $[-\pi,\pi)$ and kick parameters from $[-7/5, 7/5]$.
\item Random initialization: Each parameter in every layer uniformly sampled from $[-\pi,\pi)$.
\end{itemize}

We commence with a standard trial state $|\psi_t\rangle$ derived from an optimization of a low-dimensional matrix product state (MPS) featuring alternating bonds~\cite{Orus2019Tensor}. This MPS employs a bond dimension $\chi=2$ between the $(2j-1)$-th and the $(2j)$-th tensors, and $\chi=1$ between the $(2j)$-th and the $(2j+1)$-th tensors, for $j=1,2,\dots,n/2$. The quantum state corresponding to this MPS configuration is
\begin{equation}
\begin{aligned}
    \left|\psi_t\right\rangle=&\sum_{i_1, i_2, \ldots, i_n}\left[\prod_{j=1}^{n / 2} A^{i_{2 j-1}, i_{2 j}}\right]\\&\left[\prod_{j=1}^{n / 2-1} B^{i_{2 j}, i_{2 j+1}}\right]\left|i_1, i_2, \ldots, i_n\right\rangle,
\end{aligned}
\end{equation}
where $A^{i_{2 j-1}, i_{2 j}}$ are $2 \times 2$ tensor matrices for bonds of dimension $\chi=2$ and $B^{i_{2 j-1}, i_{2 j}}$ are scalar values for bonds of dimension $\chi=1$. These tensor parameters are tuned to minimize the trial energy $\langle \psi_t|\hat{\mathcal{H}}|\psi_t\rangle$. The resulting state $|\psi_t\rangle$ is subsequently transformed to a computational basis state $|\varphi_c\rangle$ using a shallow quantum circuit  $\hat{U}_t$, constituted by  $(n/2)$ two-qubit gates:
\begin{equation}
    \left|\varphi_c\right\rangle= \hat{U}_t|\psi_t\rangle =\prod_{j=1}^{n / 2} \hat{U}_{2 j-1,2 j}|\psi_t\rangle.
\end{equation} 
For our computations, $|\varphi_c\rangle$ is selected as the computational basis state exhibiting the highest overlap with $\left|\psi_t\right\rangle$.

\begin{figure}[t]
	\centering
	\includegraphics[width=8.2cm]{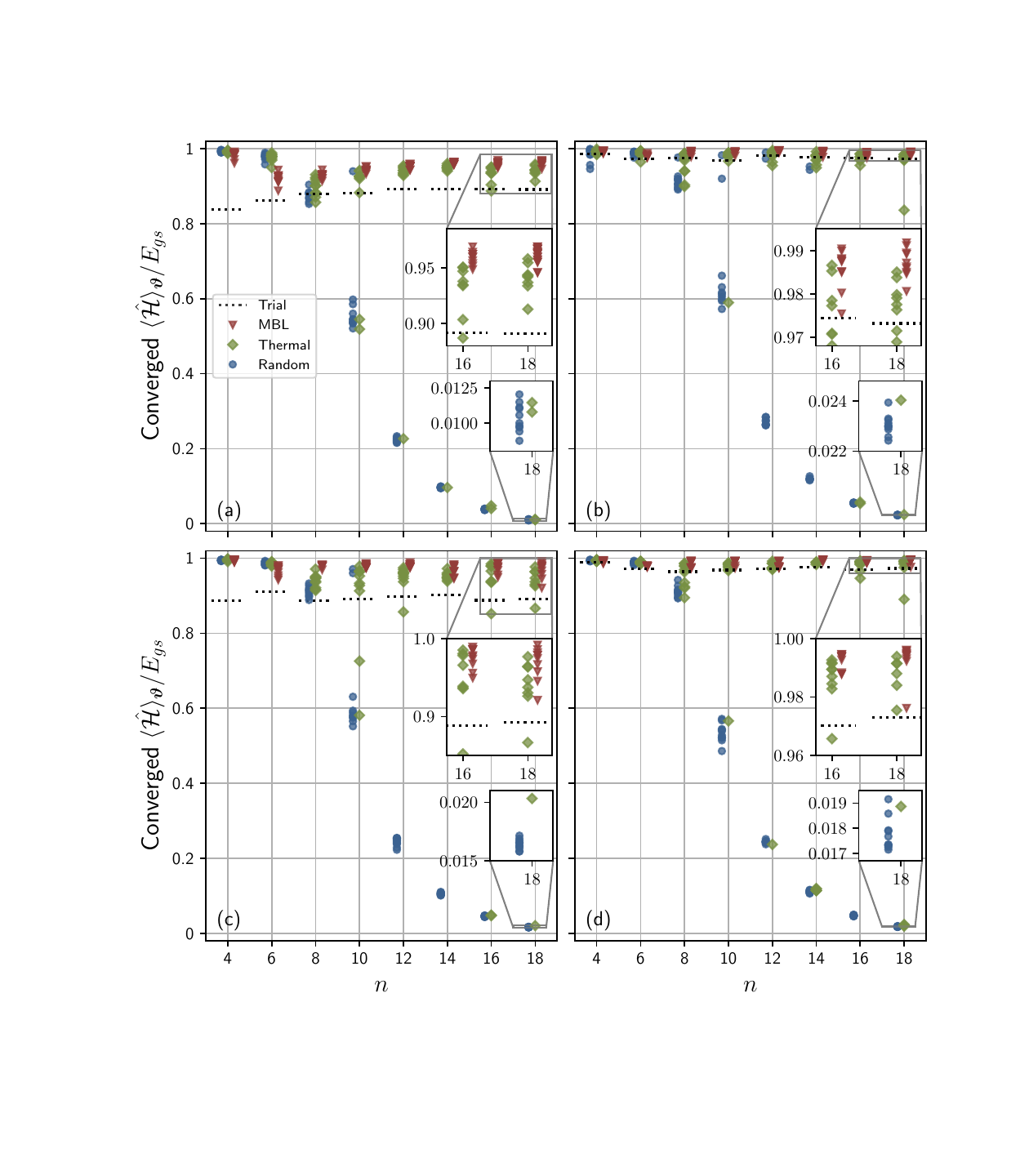}
	\caption{Converged approximation ratios from the variational quantum eigensolver for system sizes $n=4,6,8,\dots,18$ across 10 samples of initial parameters. This comparison illustrates the superiority of the MBL initialization ($W=2/5$, triangles) over thermal ($W=7/5$, diamonds) and random (circles) strategies. Black dotted lines mark the trial state energy. \textbf{(a)} $\hat{\mathcal{H}}_{\text{AA}}$ in the extended phase ($\Gamma=0, V = 1$). \textbf{(b)} $\hat{\mathcal{H}}_{\text{AA}}$ in the non-interacting Anderson localized phase ($\Gamma=0, V = 4$). \textbf{(c)} $\hat{\mathcal{H}}_{\text{AA}}$ in the ergodic phase ($\Gamma=1, V = 2$). \textbf{(d)} $\hat{\mathcal{H}}_{\text{AA}}$ in the MBL phase ($\Gamma=3, V = 6$). The insets provide a zoomed view for large systems.}
	\label{fig:OptimizedE}
\end{figure}

We proceed to the second stage: variational quantum optimization using a square circuit $\hat{U}(\boldsymbol{\vartheta})$. Initial parameters $\boldsymbol{\vartheta}$ are sampled according to the chosen initialization strategy, and optimization is performed via gradient descent with update rule \(\boldsymbol{\vartheta} \leftarrow \boldsymbol{\vartheta} - \eta \nabla \langle\hat{\mathcal{H}}\rangle_{\boldsymbol{\vartheta}}\). The learning rate is set to $\eta = 0.05$ for Fig.~\ref{fig:OptimizedE}(a–c) and $\eta = 0.01$ for (d). Optimization proceeds until the energy converges within a tolerance of 0.001 or a maximum of 1000 iterations is reached. Fig.~\ref{fig:OptimizedE} shows the ratio of the converged energy to the exact ground-state energy, averaged over 10 random initializations, as a function of system size and initialization scheme. MBL initialization consistently outperforms both thermal and random strategies, achieving approximation ratios close to unity across all Hamiltonian phases. In contrast, random initialization rapidly degrades with system size, often failing entirely for $n \gtrsim 14$, while thermal initialization performs moderately but inconsistently, occasionally producing low-energy states.

\begin{figure}[b]
	\centering
	\includegraphics[width=8.3cm]{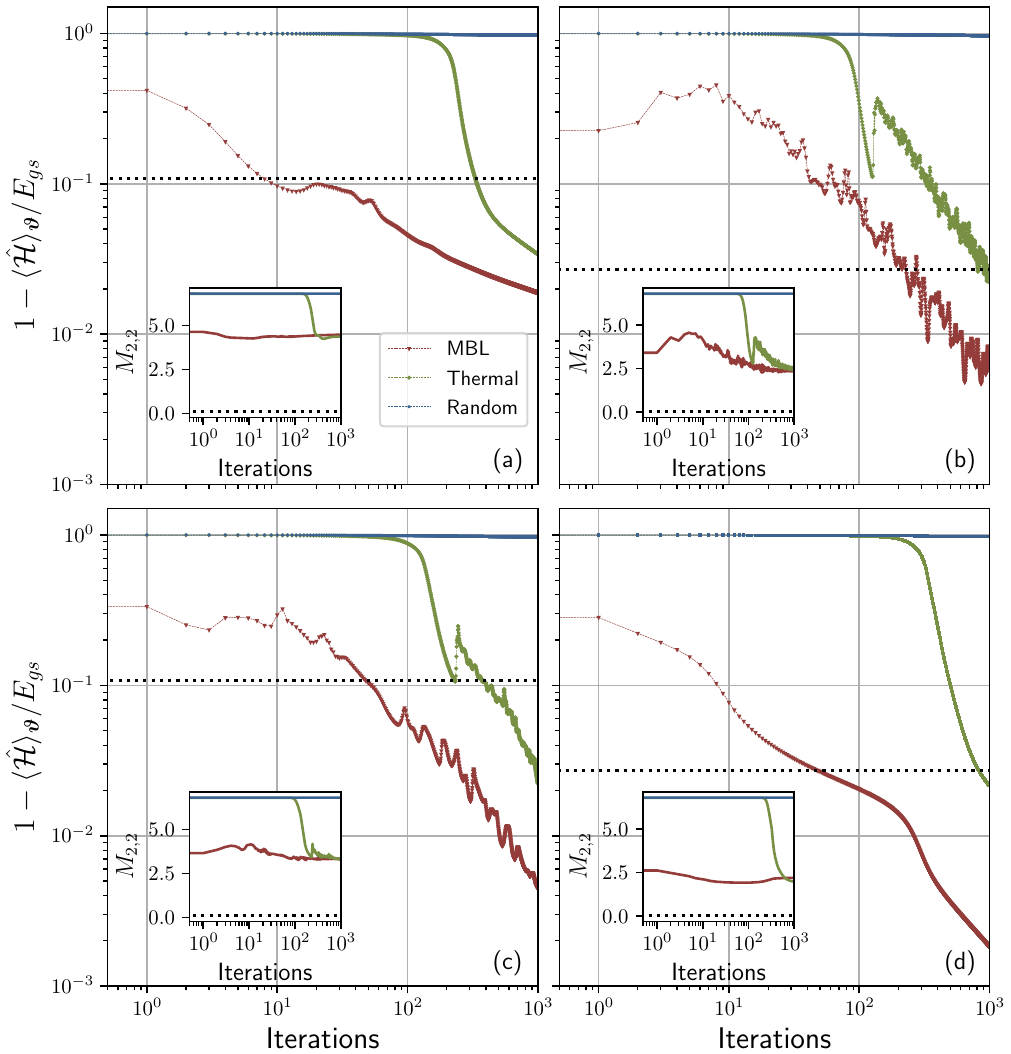}
	\caption{Relative energy error during optimization at $n=18$, showing the rapid convergence of MBL initialization. \textbf{(a)} $\hat{\mathcal{H}}_{\text{AA}}$ in the non-interacting Anderson localized phase ($\Gamma=0, V = 4$). \textbf{(b)} $\hat{\mathcal{H}}_{\text{AA}}$ in the extended phase ($\Gamma=0, V = 1$). \textbf{(c)} $\hat{\mathcal{H}}_{\text{AA}}$ in the MBL phase ($\Gamma=3, V = 6$). \textbf{(d)} $\hat{\mathcal{H}}_{\text{AA}}$ in the ergodic phase ($\Gamma=1, V = 2$). Black dotted lines indicate the energy of the trial state prior to optimization. Insets show the low-weight stabilizer Rényi entropy $M_{2,2}$ throughout the optimization, emphasizing how initialization impacts energy convergence and entropy behaviors. Each curve corresponds to a random sample of initial parameters.}
	\label{fig:OptimizedCurve}
\end{figure}

In addition to assessing initial gradients and final energies, our study is fundamentally focused on the optimization trajectory. We aim to ascertain whether MBL-based initializations can steer the optimization process towards more advantageous paths that prevent convergence to a random output state. Fig.~\ref{fig:OptimizedCurve} presents numerical results for the relative energy error and the low-weight stabilizer Rényi entropy $M_{2,2}$ throughout the optimization process for all four phases with a single instance of the initial circuit parameters. This visualization indicates that thermal initialization outperforms random initialization, which struggles to overcome barren plateaus and achieve lower energy states even after 1,000 iterations. The MBL initialization demonstrates performance superior to both, with the energy of the output state rapidly approaching that of the ground state. Furthermore, the $M_{2,2}$ trajectory underscores that when the circuit is initialized in the MBL phase, the output state remains far from its Haar benchmark throughout the optimization process.  This behavior correlates with the absence of observed barren-plateau signatures and supports the enhanced efficiency and success probability of the variational optimization under MBL-based initialization.

\section{Experimental Verification}~\label{Sec:IBMExperiment}
In this section, we detail the experimental observation of the MBL to thermalization phase transition using Floquet-initialized variational quantum circuits on the \texttt{ibm\_brisbane} quantum processor, a superconducting processor with 127 transmon qubits. Fig.~\ref{fig:IBMExperiment}(a) illustrates the architectural connectivity of the device, which includes 36 qubits of degree 3, 89 of degree 2, and 2 of degree 1. Table~\ref{tab:brisbane_parameters} lists the device's performance metrics, recorded during the experiments.

\begin{figure*}[tbh]
	\centering
	\includegraphics[width=15cm]{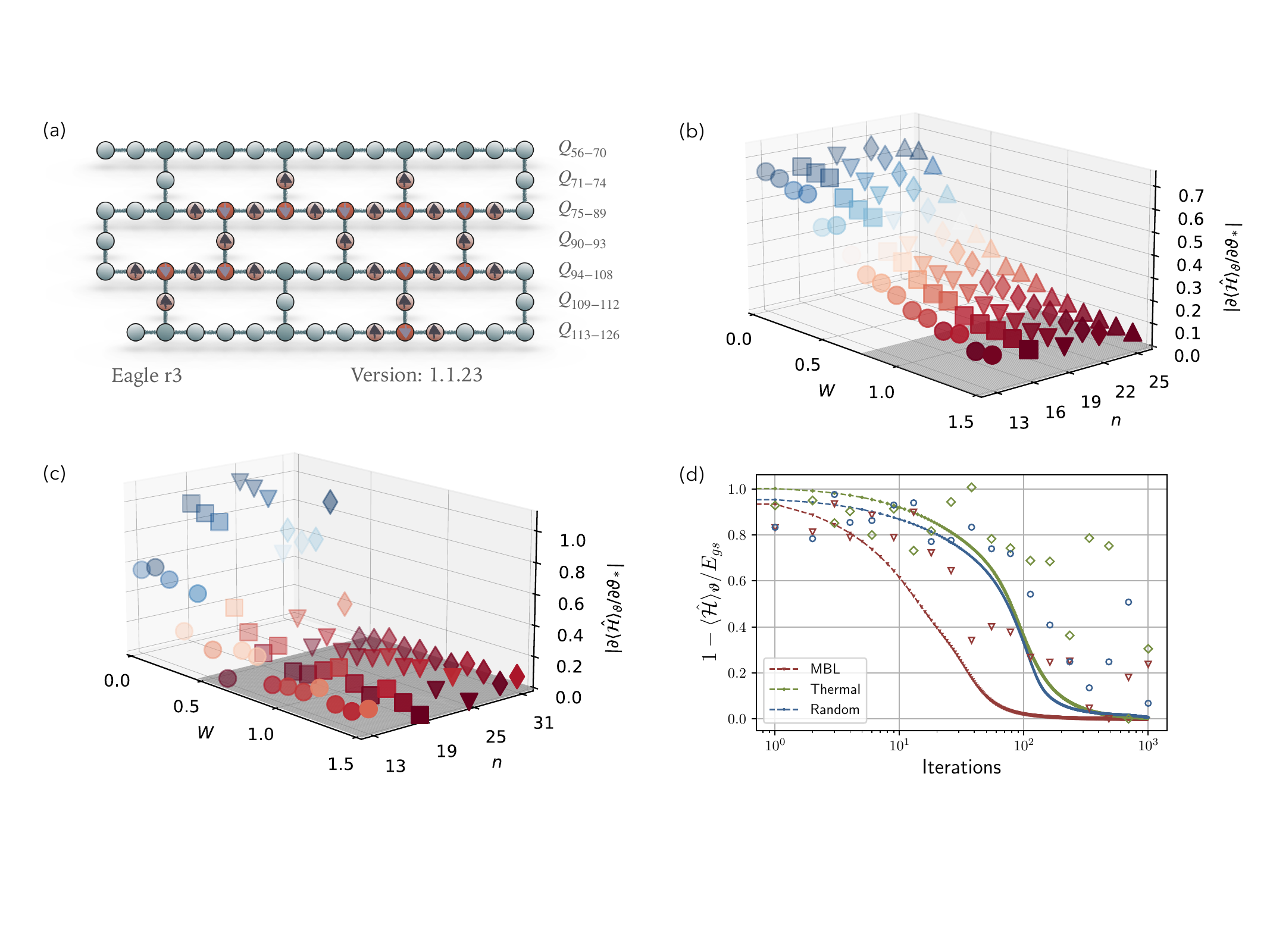}
	\caption{Experimental setup and analysis of the MBL--thermalization phase transition on IBM's \texttt{ibm\_brisbane} quantum processor. \textbf{(a)} Qubit layout of \texttt{ibm\_brisbane} alongside an intermittently coupled spin chain, represented by red circles (qubits) and connecting edges (couplings). Arrows indicate the initial alternating-spin Néel configuration. \textbf{(b)} Simulation results of the observed gradients of the Hamiltonian's expectation value as a function of the kick strength $W$ for system sizes $n = 13, 16, 19, 21, 25$. \textbf{(c)} Experimental results of the observed gradients for system sizes $n = 13, 19, 25, 31$, with gradients similarly demonstrating decay patterns as $W$ increases, indicating a phase transition at an experimentally adjusted point $W^\star \approx 1/2$. The experimental critical point is slightly shifted due to hardware noise. \textbf{(d)} Relative error during optimization shown numerically (lines) and experimentally (hollow markers) for system size $n=19$.}
	\label{fig:IBMExperiment}
\end{figure*}

\begin{table}[ht]
\centering
\renewcommand{\arraystretch}{1.5}
\begin{tabular}{l|l}
\hline
Parameter & Median Value \\
\hline
\# of qubits & 127\\
$T_1$ & \(232.17 \mu \mathrm{s}\)  \\
$T_2$ & \(
153.64 \mu\mathrm{s}\)   \\
ECR error & $7.802\times 10^{-3}$ \\
SX error & $2.428\times 10^{-4}$ \\
Readout error & $1.370\times 10^{-2}$ \\
Error per layered gate& 1.9\%  \\
\hline
\end{tabular}
\caption{Performance metrics of the \texttt{ibm\_brisbane} quantum processor. Critical parameters include the number of qubits, coherence times ($T_1$ and $T_2$), error rates for echoed cross-resonance (ECR) gates and $\sqrt{X}$ (SX) gates, readout error, and the average error per gate in a layered circuit configuration.}
\label{tab:brisbane_parameters}
\end{table}

\begin{table}[ht]
\centering
\renewcommand{\arraystretch}{1.5}
\begin{tabular}{l|l}
\hline
\textbf{Setup} & \textbf{Qubits Used} \\
\hline
$n=13$ & $Q_{73}$, $Q_{82\text{--}86}$, $Q_{92}$, $Q_{101\text{--}105}$, $Q_{111}$. \\
$n=19$ & $Q_{72}$, $Q_{78\text{--}84}$, $Q_{91}$, $Q_{92}$, $Q_{97\text{--}99}$, \\
 &  $Q_{101\text{--}105}$, $Q_{111}$. \\
$n=25$ & $Q_{72}$, $Q_{73}$, $Q_{78\text{--}88}$, $Q_{91\text{--}93}$, $Q_{97\text{--}99}$,  \\
 & $Q_{103\text{--}107}$, $Q_{111}$. \\
$n=31$ & $Q_{72}$, $Q_{73}$, $Q_{78\text{--}88}$, $Q_{91\text{--}93}$, $Q_{95\text{--}99}$, \\ &  $Q_{103\text{--}107}$, $Q_{109}$, $Q_{111}$,  $Q_{121\text{--}123}$.\\
\hline
\end{tabular}
\caption{Detailed qubits usage for different configurations with grouped subscript ranges}
\label{tab:qubits_usage}
\end{table}

For clarity in our experimental setup, we focus on a specialized configuration of $\hat{\mathcal{H}}_{\text{AA}}$ with parameters $J=1$, $V = \Gamma = -2J$, $\alpha = 0$, and $\phi = \pi$. This setup corresponds to the ferromagnetic Heisenberg model, governed by the Hamiltonian:
\begin{equation}
    \begin{aligned}
    \hat{\mathcal{H}}_{\text{FH}} = &-\frac{1}{2} \sum_{\substack{(i, j) \in \mathrm{edges}(G)}} \left(\hat{\sigma}_i^x\hat{\sigma}_{j}^x + \hat{\sigma}_i^y\hat{\sigma}_{j}^y + \hat{\sigma}_i^z\hat{\sigma}_{j}^z\right).
\end{aligned}
\end{equation}
This model is implemented on an intermittently coupled spin chain of $n$ qubits, where alternate qubits are linked to an additional, isolated qubit, leading to non-uniform interactions across all spins. This configuration aligns well with the connectivity graph of the \texttt{ibm\_brisbane} quantum processor. Notably, the circuit layout is dynamically optimized by the Qiskit compiler for $W=1/2$ to enhance performance in practical experiments. In one of our experiments, this configuration is applied to 31 qubits, depicted in Fig.~\ref{fig:IBMExperiment}(a) with red circles representing qubits and connecting edges. The ground states of \(\hat{\mathcal{H}}_{\text{FH}}\) are reachable with a shallow circuit, allowing us to isolate the gradient-collapse signal at the MBL--thermalization crossover without conflating it with state-preparation complexity.

Before the experiment, we numerically simulate the system using a Floquet initialization strategy where the rotation angles within $\hat{\mathcal{R}}_{x}$, $\hat{\mathcal{R}}_{xx}$, and $\hat{\mathcal{R}}_{yy}$ are uniformly sampled from the range $[-W,W]$. Our analysis focuses on the Pauli-$XX$ rotation at the center of the first layer of the spin chain, which has been verified to accurately represent the global trend observed across the entire gradient profile. A large gradient at a representative parameter (\(\vartheta_*\)) empirically correlates with trainability of the full landscape. Considering our specific interest in observing phase transitions marked by the emergence and disappearance of significant gradients, we opt for a naive trial state in an alternating-spin Néel configuration, denoted as $|\psi_t \rangle = |\varphi_c\rangle = \left|\uparrow\right\rangle\left|\downarrow\right\rangle\cdots \left|\uparrow\right\rangle\left|\downarrow\right\rangle$, visually represented in Fig.~\ref{fig:IBMExperiment}(a). We measure the gradient of the Hamiltonian's expectation value relative to $\vartheta_*$ by executing each circuit configuration 4096 times. The circuit depth is set to $D=\lfloor (n-1)/6\rfloor$ for a system with $n$ qubits. The average absolute gradient, depicted in Fig.~\ref{fig:IBMExperiment}(b) for 200 samples of initial parameters, displays a clear phase transition from MBL to thermalization around a point of $W \approx 4/5$.

To initiate the experiment efficiently with limited resources, we employ a representative specialized sampling protocol for the initial Floquet circuit parameters, which includes all rotation angles within $\hat{\mathcal{R}}_{x}$, $\hat{\mathcal{R}}_{z}$, $\hat{\mathcal{R}}_{xx}$, $\hat{\mathcal{R}}_{yy}$, and $\hat{\mathcal{R}}_{zz}$. These parameters are generated using the zeros of the Chebyshev polynomial of the first kind, specifically:
\begin{equation}
\vartheta_j = \pi\cos \left(\frac{(2 j-1) \pi}{2( m_s +  m_k)}\right),
\end{equation}
for  $j = 1, 2, \ldots,  m_s +  m_k$, with $ m_s = 2n$ and $ m_k = 3n$, so that all angles initially lie in $[-\pi,\pi]$. We arrange these zeros in a staggered pattern, placing positive values at even indices and negative values at odd indices within the rotation-angle array. The angles assigned to $\hat{\mathcal{R}}_{z}$ and $\hat{\mathcal{R}}_{zz}$ are then kept fixed and play the role of the static steady disorder, while the angles in $\hat{\mathcal{R}}_{x}$, $\hat{\mathcal{R}}_{xx}$, and $\hat{\mathcal{R}}_{yy}$ are further rescaled by the kick strength $W$ to ensure they fall within the range $[-W, W]$. Although this deterministic sampling deviates from the uniform randomness assumed in the theoretical analysis, it provides a representative set of disordered parameters that we expect to capture the essential features of the MBL--thermalization transition. Readout error mitigation is applied to enhance measurement accuracy, whereas gate error mitigation is intentionally omitted to preserve the intrinsic dynamics of the system.

The experiment is conducted across system sizes of $n = 13, 19, 25, 31$ using noisy physical qubits, as detailed in Table~\ref{tab:qubits_usage}. As shown in Fig.~\ref{fig:IBMExperiment}(c), the observed absolute gradient values display a gradual decline in magnitude as the kick parameter $W$ is increased, signaling a transition towards the thermal phase. Despite noise effects, an experimental critical transition point emerges near \(W=1/2\). Below this threshold, gradient magnitudes remain partially preserved. Notably, this experimental critical point is lower than the numerically predicted value $W=4/5$, highlighting the impact of hardware noise on shifting the observed critical threshold. However, the experimental results still corroborate the general theoretical trend. 

To more specifically illustrate the impact of the gradient on the optimization path, we plot the relative error throughout the optimization process for a system with $n=19$ in Fig.~\ref{fig:IBMExperiment}(d). It displays both numerically simulated energy and experimental results (that have been rescaled for error mitigation) for several selected steps, using identical parameter sets. The circuit parameters are optimized classically. Compared to thermal and random initializations, the optimization employing MBL initialization demonstrates a rapid convergence towards the ground state. The improvement here is modest because this Hamiltonian is intentionally simple—our aim is to validate gradient revival experimentally, not to showcase a hard energy-minimization task. More demanding benchmarks, including assessments of classical MPS limitations on long-range models, comparisons with close-to-identity initialization to highlight escape from local minima, and robustness tests against trial-state perturbations and finite-shot noise, appear in Appendix~\ref{appendix:benchmark}. These extensions demonstrate the protocol's practical advantages in challenging scenarios beyond the simple Heisenberg case.

\section{Conclusions and Outlook}~\label{Sec:Conclusions}
The primary motivation behind this research was to address the significant challenge of barren plateaus in variational quantum simulation, a phenomenon that considerably impedes the scalability and efficiency of quantum optimization. This investigation explored the transition from many-body localization (MBL) to thermalization within a universal variational quantum circuit framework. We introduced an MBL-based strategy aimed at mitigating excessive entanglement and thereby circumventing barren plateaus in variational quantum optimization. The application of this protocol led to marked enhancements in energy convergence for the Aubry-André model across various phases, promoting more efficient optimization processes. A significant aspect of our study was the experimental validation of the phase transition on the 127-qubit \texttt{ibm\_brisbane} quantum processor through gradient computation. Moreover, this research innovates by utilizing barren plateaus as a novel metric for identifying the critical point of the MBL--thermalization phase transition. While our findings offer empirical evidence and establish a conditional bound on gradient magnitudes, we cannot rule out the emergence of barren plateaus at later stages of training.

Our results contribute to a deeper understanding of how MBL can mitigate barren plateaus in variational quantum optimization, consistent with insights from Ref.~\cite{Cerezo2023Does} regarding the relationship between barren plateaus and classical simulability. This alignment is particularly evident in the early stages of optimization, given the classical simulability of MBL systems~\cite{Vidal2024Efficient}. Despite this, our methodology is intended to maintain substantial
optimization gradients, even during a possible crossover to a thermal phase, thereby potentially enabling operation in regimes previously considered classically intractable. Warm-start analyses warn, however, that training will still fail unless one can initialize increasingly close to the relevant region of attraction as the system grows~\cite{Mhiri2025Unifying}. Furthermore, Ref.~\cite{Lerch2024Efficient} shows that a polynomial-time classical surrogate of such a local patch can, in principle, be built after a moderate quantum data-acquisition stage. Whether such surrogates remain practical in our setting is an open—and important—question for future work.

This work opens several promising avenues for future research. One potential direction involves applying the MBL-based initialization protocol to quantum machine learning tasks, such as quantum generative adversarial networks~\cite{Zoufal2019Quantum}. Utilizing the unique properties of MBL systems could potentially lead to more efficient and trainable quantum machine learning algorithms capable of generating distributions beyond the capabilities of classical methods. Further, exploring the interplay between MBL and other aspects of quantum circuits, such as measurement-induced phase transitions~\cite{Roeland2023Measurement} and information scrambling~\cite{Mi2021Information}, could yield valuable insights into quantum dynamics. Additionally, employing MBL in specialized variational quantum protocols, including QAOA variants~\cite{Yu2022Quantum, Yu2023Solution, Yu2025Warm, Sun2023Improving} and the variational simulation of time evolution~\cite{Yuan2019theoryofvariational, Benedetti2021Hardware}, presents further exploration possibilities.

However, several open questions and potential enhancements persist within our framework. One critical issue is the selection of optimal trial states whose entanglement structures align closely with those of the target system's ground state. This alignment is crucial for enhancing the efficiency of quantum optimization but remains a complex challenge. While trial states from matrix product states can be compiled into log-depth shallow circuits~\cite{Ran2020Encoding, Rudolph2023Decomposition, Malz2024Preparation}, an alternative is the use of measurement-driven protocols to prepare highly entangled states in constant-depth, thus circumventing light-cone constraints~\cite{Baumer2024Efficient, Sahay2024Finite, Cao2025Measurement}. Another significant consideration is the quantitative assessment of hardware noise effects on the critical points of MBL and thermalization phases. Fully understanding how this factor influences the transition points and the overall dynamics of variational quantum circuits is essential for enhancing their practical effectiveness. In parallel, MBL-based initialization is compatible with complementary post-processing strategies that mitigate algorithmic errors in quantum optimization, such as energy-extrapolation schemes~\cite{Cao2022Mitigating}. Additionally, integrating MBL initialization with advanced techniques such as joint Bell measurements or circuit parallelism could significantly accelerate optimization processes~\cite{Cao2024Accelerated, Mineh2023Accelerating}.

\acknowledgments
We are grateful to Luya Lou, Paul Skrzypczyk, Yunlong Yu, Tianyi Hao, Qihao Guo, Zheng-Hang Sun, B. \"{O}zguler, and M. Vavilov for insightful discussions. 
This work was supported by the Theory of Quantum Matter Unit, OIST. 
A part of this research was conducted while visiting the Okinawa Institute of Science and Technology (OIST) through the Theoretical Sciences Visiting Program (TSVP). 
We acknowledge the use of IBM Quantum services for this work.

\onecolumngrid
\bibliographystyle{quantum}
\bibliography{main}

\onecolumngrid
\begin{appendices}

\section{Proof of Circuit Universality}\label{obproof:universality}
\begin{proposition}[Circuit universality]
The sequential layered variational circuit in Fig.~\ref{fig:schematic} is universal: for any \( U \in \mathrm{SU}(2^n) \) and any \( \varepsilon > 0 \), there exist parameters \( \boldsymbol{\vartheta} \) such that the unitary \( \hat{U}(\boldsymbol{\vartheta}) \) satisfies
\begin{equation}
\left\| \hat{U}(\boldsymbol{\vartheta}) - U \right\|_\infty \le \varepsilon,
\end{equation}
with circuit depth exponential in \( n \) and polylogarithmic in \( 1/\varepsilon \). Moreover, if \( U \) admits a nearest-neighbor decomposition with \( \mathrm{poly}(n, \log(1/\varepsilon)) \) layers of two-qubit gates, the same accuracy can be achieved with depth \( D = \mathrm{poly}(n, \log(1/\varepsilon)) \).
\end{proposition}

\begin{figure}[b]
	\centering
	\includegraphics[width=6.5cm]{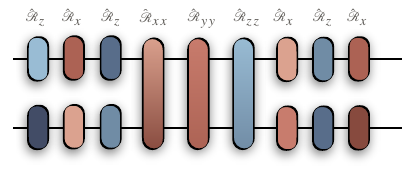}
	\caption{Cartan's KAK decomposition of arbitrary double-qubit unitary $U \in \operatorname{SU}(4)$. The diagram illustrates a constructive decomposition of a general two-qubit unitary into 15 parameterized gates: local Euler rotations and a central block of $XX$, $YY$, and $ZZ$ interactions.}
	\label{fig:KAK}
\end{figure}

\begin{proof}
    Consider any pair of adjacent qubits, denoted as $Q_j$ and $Q_k$. We can construct a layout with 15 parameterized gates in three consecutive layers, as depicted in Fig.~\ref{fig:KAK}. It is known that any single-qubit unitary operation can be achieved through consecutive rotations of Pauli-$Z-X-Z$ or Pauli-$X-Z-X$, up to a global phase. Moreover, the operation
\begin{equation}
\begin{aligned}
&\exp\left(-i\frac{\vartheta_{xx}}{2} \hat{\sigma}_j^x \hat{\sigma}_k^x\right)
\exp\left(-i\frac{\vartheta_{yy}}{2} \hat{\sigma}_j^y \hat{\sigma}_k^y\right) \exp\left(-i\frac{\vartheta_{zz}}{2} \hat{\sigma}_j^z \hat{\sigma}_k^z\right)
\end{aligned}
\end{equation}
    is equivalent to the operation
    \begin{equation}
        \exp(-i\frac{\vartheta_{xx}}{2} \hat{\sigma}_j^x\hat{\sigma}_k^x-i\frac{\vartheta_{yy}}{2} \hat{\sigma}_j^y\hat{\sigma}_k^y-i\frac{\vartheta_{zz}}{2} \hat{\sigma}_j^z\hat{\sigma}_k^z).
    \end{equation}
    Consequently, we can employ Cartan's KAK decomposition~\cite{Tucci2005An}:
    \begin{equation}
    \begin{aligned}
        \hat{U}(\boldsymbol{\vartheta}) = &\hat{U}_s(\vartheta_1,\vartheta_2,\vartheta_3)\otimes \hat{U}_s(\vartheta_4,\vartheta_5,\vartheta_6) \exp(-i\frac{\vartheta_7}{2} \hat{\sigma}_j^x\hat{\sigma}_k^x-i\frac{\vartheta_8}{2} \hat{\sigma}_j^y\hat{\sigma}_k^y-i\frac{\vartheta_9}{2} \hat{\sigma}_j^z\hat{\sigma}_k^z)\\
        &\cdot \hat{U}_s(\vartheta_{10},\vartheta_{11},\vartheta_{12})\otimes \hat{U}_s(\vartheta_{13},\vartheta_{14},\vartheta_{15}),
    \end{aligned}
    \end{equation}
    where $\hat{U}_s(\vartheta_\alpha,\vartheta_\beta,\vartheta_\gamma)$ represents a 3-parameter universal single-qubit unitary acting on $Q_j$ or $Q_k$. Any two-qubit unitary operation $\hat{U} \in \operatorname{SU}(4)$ can be decomposed into the KAK form. As a result, we can implement arbitrary 1-qubit and 2-qubit gates, such as the universal gate set $\{\text{CNOT}, \text{H}, \text{S}, \text{T}\}$. 
    
    A nearest-neighbour SWAP network of depth \(\mathcal{O}(n)\) brings any two qubits adjacent. Applying the Solovay–Kitaev theorem~\cite{Kitaev1997Quantum} then ensures an approximation within \(\varepsilon\), with an overall circuit depth exponential in \(n\) and polylogarithmic in \(1/\varepsilon\).

    If \(U\) already admits a nearest-neighbour decomposition with \(\mathrm{poly}(n, \log(1/\varepsilon))\) layers, the same reasoning directly yields depth \(D = \mathcal{O}(\mathrm{poly}(n,\log(1/\varepsilon)))\).
\end{proof}

\section{Level-statistics confirmation of CUE}\label{Appendix:CUE}
To determine the appropriate random-matrix ensemble characterizing the thermal phase of our Floquet-initialized circuits, we conduct a level-statistics analysis of the quasienergies. Starting from Poisson (MBL) statistics at small kick strength $W$, the spectra change towards the circular unitary ensemble (CUE) as $W$ increases, signaling ergodicity and level repulsion.

\begin{figure}[b]
\centering
\includegraphics[width=13cm]{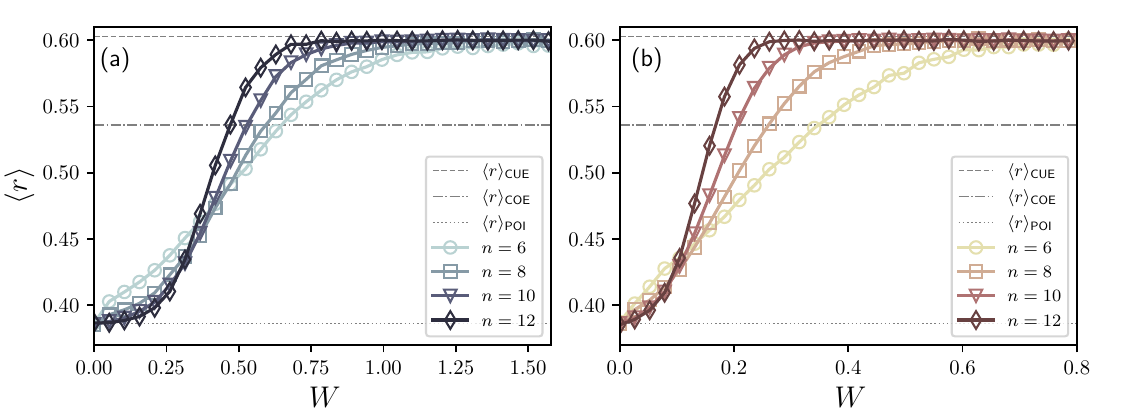}
\caption{Disorder‑averaged adjacent‑gap ratio \(\langle r \rangle\) versus kick strength \(W\) for system sizes $n=6,8,10,12$. \textbf{(a)} One-dimensional ring connectivity. \textbf{(b)} Circulant graph $\mathrm{Ci}_{n}(1,2)$. Horizontal dashed lines (from top to bottom) mark the universal CUE ($\langle r\rangle\simeq0.603$), COE ($\langle r\rangle\simeq0.536$), and Poisson ($\langle r\rangle\simeq0.386$) benchmarks. The crossover of the numerical curves from the Poisson plateau at small $W$ to the CUE plateau at large $W$ signals the many-body-localization-to-thermalization transition; data for larger $n$ (darker shades) sharpen this crossover, consistent with the main-text diagnostics. Each point is obtained by averaging over 200 independent disorder realizations.}
\label{fig:level-stats}
\end{figure}

The principal diagnostic is the disorder‑averaged adjacent‑gap ratio
\begin{equation}
\langle r\rangle
:= \bigl\langle
\frac{\min(\delta_i,\delta_{i+1})}{\max(\delta_i,\delta_{i+1})}
\bigr\rangle
\end{equation}
where $\delta_i$ represents the spacing between consecutive quasienergy levels (eigenphases of the Floquet unitary). For Poisson statistics, $\langle r \rangle \approx 0.386$; for the circular orthogonal ensemble (COE), $\langle r \rangle \approx 0.536$; and for the circular unitary ensemble (CUE), $\langle r \rangle \approx 0.603$.

For $n=6,8,10,12$ qubits we diagonalized the Floquet unitary in both the one‑D ring and the circulant graph $\mathrm{Ci}_n(1,2)$ topologies.  
Fig.~\ref{fig:level-stats} shows $\langle r\rangle$ versus~$W$:
all curves cross over from the Poisson plateau at small~$W$ to the CUE plateau at large~$W$, with the transition sharpening as $n$ grows. In the thermal regime we obtain $\langle r\rangle\simeq0.60$, fully consistent with CUE and incompatible with COE. This behaviour is expected because our ansatz employs
non‑commuting rotations with complex phases, thereby breaking
time‑reversal symmetry and generating the full Lie algebra
$\mathfrak{su}(2^n)$. In this work, we pinpoint the MBL--thermalization phase transition using metrics such as the inverse participation ratio, entanglement entropy, and low-weight stabilizer Rényi entropy (as detailed in the main text), rather than level statistics. These measures are more closely linked to variational quantum algorithm (VQA) trainability, entanglement scaling, and localization properties, and prove computationally efficient for the studied system sizes, circumventing the requirement for full spectral diagonalization.

Even if the thermal phase were governed by COE (e.g., for real-valued gates yielding unitaries in the orthogonal group $\mathcal{O}(d)$ with $d=2^n$), our $t$-design-based diagnostics would hold. The moments of COE and CUE are asymptotically equivalent, with discrepancies manifesting only as sub-leading terms in the Weingarten expansion~\cite{Collins2006Integration}. Concretely, for fixed order $t$ and large $d$,
\begin{equation}
\begin{aligned}
   &\int_{\mathrm O(d)}  O^{\otimes t}  \otimes O^{\dagger\otimes t} d\mu_{\rm O}(O)
    = 
   \int_{\mathrm U(d)}  U^{\otimes t}  \otimes U^{\dagger\otimes t} d\mu_{\rm U}(U)
   + \mathcal O  \left(\frac{1}{d}\right),
\end{aligned}
\end{equation}
where $\mu_{\rm O}$ and $\mu_{\rm U}$ denote the respective Haar measures. Thus, applying a COE unitary to a fixed computational basis state generates an orthogonal $t$-design, with differences vanishing for large systems and absent from our normalized metrics.

\section{Proof of Theorem~\ref{theorem:IPR}}\label{theoremproof:IPR}

\begin{proof}
Let $|\psi_\mathrm{in}\rangle$ be an arbitrary $n$-qubit input state, and consider the parametrized quantum state
$|\psi(\boldsymbol{\vartheta})\rangle=\hat U(\boldsymbol{\vartheta})|\psi_\mathrm{in}\rangle$.
Since the ensemble $\{\hat U(\boldsymbol{\vartheta})\}_{\vartheta\in\Theta}$ forms a unitary
$t$-design, its $t$-fold moments reproduce those of the Haar measure.
Consequently, for every basis vector $|\beta\rangle$, we have
\begin{equation}\label{eq:haar-moment}
  \int_\Theta
|\langle\beta|\psi(\boldsymbol{\vartheta})\rangle|^{2t}d\boldsymbol{\vartheta}
  =
  \mathbb E_{|\phi\rangle\sim\mathrm{Haar}}
     \bigl[|\langle\beta|\phi\rangle|^{2t}\bigr].
\end{equation}
It is known that, for a Haar-random state \(|\phi\rangle\), the moment on the right-hand side equals the inverse of the dimension of the totally symmetric subspace \(\vee^t \mathbb{C}^{2^n}\)~\cite{Harrow2013Church, Mele2024Introduction}. Specifically,
\begin{equation}
  \dim \bigl(\vee^{t} \mathbb C^{2^{n}}\bigr)=\binom{2^{n}+t-1}{t}
\end{equation}
and
\begin{equation}
\begin{aligned}
  \mathbb E_{|\phi\rangle\sim{\rm Haar}}
     \bigl[|\langle\beta|\phi\rangle|^{2t}\bigr]
  =&\mathbb E_{|\phi\rangle\sim{\rm Haar}}
     \bigl[\langle\beta|^{\otimes t}|\phi\rangle^{\otimes t}
            \langle\phi|^{\otimes t}|\beta\rangle^{\otimes t}\bigr]\\
  =&\langle\beta|^{\otimes t}
      \bigl[\mathbb E_{|\phi\rangle\sim{\rm Haar}}
            \bigl(|\phi\rangle^{\otimes t}  \langle\phi|^{\otimes t}\bigr)\bigr]
     |\beta\rangle^{\otimes t} \\=& 
     \frac{\langle\beta|^{\otimes t}\Pi_{\rm sym}|\beta\rangle^{\otimes t}}
          {\dim  \bigl(\vee^{t}  \mathbb C^{2^{n}}\bigr)}
      \\=&\frac{1}{\dim  \bigl(\vee^{t}  \mathbb C^{2^{n}}\bigr)},
\end{aligned}
\end{equation}
where \(\Pi_{\rm sym}\) denotes the projector onto the totally symmetric subspace. Substituting this expression into Eq.~\eqref{eq:haar-moment} and summing over the complete computational basis $\{|\beta_j\rangle\}_{j=1}^{2^{n}}$ we
obtain
\begin{equation}
\begin{aligned}
\int_\Theta \operatorname{IPR}_t \bigl(|\psi(\boldsymbol{\vartheta})\rangle\bigr)
   d\boldsymbol{\vartheta}
   =
   \frac{2^{n}}{\dim \bigl(\vee^{t} \mathbb C^{2^{n}}\bigr)}
   =
   \frac{2^{n}! \,  t!}{\bigl(t+2^{n}-1\bigr)!}.
\end{aligned}
\end{equation}
Setting $t=2$ immediately yields
\begin{equation}
\dim \bigl(\vee^{2} \mathbb C^{2^{n}}\bigr)
     =\binom{2^{n}+1}{2}
     =\frac{2^{n}(2^{n}+1)}{2}.
\end{equation}
Hence
\begin{equation}
    \int_{\Theta} \text{IPR}_2(|\psi(\boldsymbol{\vartheta})\rangle)d\boldsymbol{\vartheta} = \frac{2}{(2^n+1)}.
\end{equation}
\end{proof}

\section{State 2-Design and Circuit Expressibility}\label{Sec: Frame Potential}
As previously noted in Section~\ref{Sec:MBLTransition}, although the inverse participation ratio ($\text{IPR}_t$) can suggest whether a quantum state ensemble deviates from forming a $t$-design, it is not solely sufficient for such a determination. To provide further evidence that, in the thermal phase, the output-state ensemble approaches a state $t$-design (with respect to the Haar measure on pure states), we investigate the frame potential, defined for $t$-th order moments as:
\begin{equation}
\mathcal{F}_t = \int_{\Theta} \int_{\Theta} \left|\left\langle \psi(\boldsymbol{\vartheta}_1) \mid \psi(\boldsymbol{\vartheta}_2)\right\rangle\right|^{2t} d\boldsymbol{\vartheta}_1 d\boldsymbol{\vartheta}_2.
\end{equation}
An ensemble qualifies as a state $t$-design if its frame potential matches the Haar measure frame potential, known as the Welch bounds:
\begin{equation}
\mathcal{F}_t = \mathcal{F}^{(\mathrm{Haar})}_t := \frac{1}{\dim \bigl(\vee^{t} \mathbb C^{2^{n}}\bigr)} = \frac{1}{\binom{t + 2^n - 1}{t}}.
\end{equation}

\begin{figure*}[b]
    \centering
    \includegraphics[width=16cm]{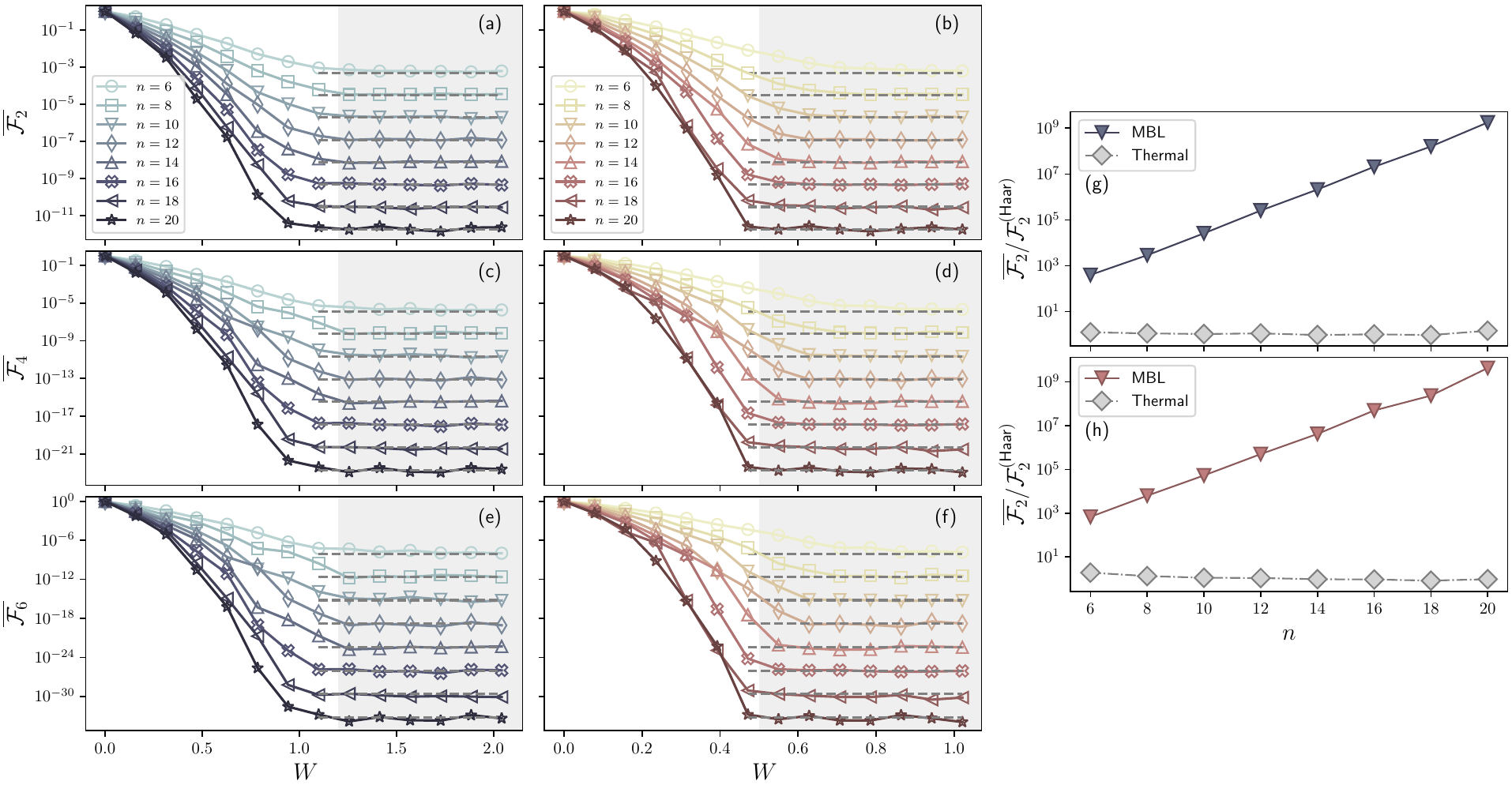}
    \caption{Frame potentials of output states across different orders as a function of kick strength $W$ for various lattice topologies, demonstrating how close the ensembles approach state $t$-designs. \textbf{(a, c, e)} display results for the 1D ring topology. \textbf{(b, d, f)} show results for the  $\mathrm{Ci}_{n}(1,2)$ lattice topology. Each panel presents frame potentials of different orders: \textbf{(a, b)} second-order, \textbf{(c, d)} fourth-order, and \textbf{(e, f)} sixth-order. Data averaged over 100 samples, with grey dashed lines indicating the Welch bounds~\cite{Datta2012Geometry}. \textbf{(g)} and \textbf{(h)} illustrate the ratios between the computed frame potentials and the corresponding Haar-random frame potentials for the 1D ring and the $\mathrm{Ci}_{n}(1,2)$ lattice, respectively.}
    \label{fig:FP transition}
\end{figure*}

For a state 2-design in particular, this criterion simplifies to:
\begin{equation}
\mathcal{F}_2 = \mathcal{F}^{(\mathrm{Haar})}_2 := \frac{2}{2^n(2^n + 1)}.
\end{equation}
We calculate the average frame potentials $\mathcal{F}_2$, $\mathcal{F}_4$, and $\mathcal{F}_6$ across various values of $W$, and present these findings in Fig.~\ref{fig:FP transition} (a-f). In the thermal phase, these averaged frame potentials accurately reflect their theoretical Haar values, substantiating the ensemble's transition to state designs.

Additionally, the ratio between $\mathcal{F}_2$ and $\mathcal{F}^{(\mathrm{Haar})}_2$ is indicative of the circuit's state-dependent expressibility~\cite{Holmes2022Expressibility}, defined as:
\begin{equation}
\begin{aligned}
\varepsilon_t(|\psi_t\rangle\langle\psi_t|) := 
\left\| \int_{\Theta} \hat{U}(\boldsymbol{\vartheta})^{\otimes t}
(|\psi_t\rangle\langle\psi_t|^{\otimes t}) 
\hat{U}^\dagger(\boldsymbol{\vartheta})^{\otimes t}  d\boldsymbol{\vartheta} \right. 
\left. - \int_{\mathrm{U}(2^n)} \hat{U}^{\otimes t} 
(|\psi_t\rangle\langle\psi_t|^{\otimes t}) 
\hat{U}^{\dagger \otimes t}  d\mu_{\mathrm H}(\hat{U}) \right\|_2.
\end{aligned}
\end{equation}
Here, our trial state $|\psi_t\rangle$ is sampled from the computational basis. Larger ratios suggest reduced expressibility of the ansatz, whereas ratios approaching 1 indicate maximal expressibility associated with strong barren plateaus. Fig.~\ref{fig:FP transition} (g) shows these ratios for the MBL phase ($W = \pi/10$) and the thermal phase ($W = 9\pi/20$) in the 1D ring topology, and Fig.~\ref{fig:FP transition} (h) for the MBL phase ($W = \pi/20$) and the thermal phase ($W = 9\pi/40$) in the $\mathrm{Ci}_{n}(1,2)$ lattice topology. In the MBL phase, these ratios exponentially increase with the system size $n$, indicating limited expressibility, whereas in the thermal phase, they remain constant near 1, indicative of maximal expressibility.

\section{Proof of Theorem~\ref{theorem:SE}}\label{theoremproof:SE}

Throughout this appendix we denote \(N=2^n\) and always work with the \(q:=2t\)‑th moment of overlaps.

\begin{proof}
Let \(\hat O\) be a traceless Hermitian operator with eigenvalues \(\pm\lambda\). Choose an
eigenbasis
\(\bigl\{|p_j\rangle,|m_j\rangle\bigr\}_{j=1}^{N/2}\) such that
\begin{equation}
    \hat O|p_j\rangle=+\lambda|p_j\rangle,
\end{equation}
and
\begin{equation}
\hat O|m_j\rangle=-\lambda|m_j\rangle.
\end{equation}
For a pair of multi–indices  
\(\boldsymbol\alpha=(\alpha_1,\dots,\alpha_{q})\in\{1,\dots,N/2\}^{q}\)  
and  
\(\boldsymbol\sigma=(\sigma_1,\dots,\sigma_{q})\in\{+1,-1\}^{q}\)  
define
\begin{equation}
  |v_\ell\rangle:=
  \begin{cases}
    |p_{\alpha_\ell}\rangle,&\sigma_\ell=+1,\\[2pt]
    |m_{\alpha_\ell}\rangle,&\sigma_\ell=-1,
  \end{cases}
\quad
  \sgn(\boldsymbol\sigma):=\prod_{\ell=1}^{q}\sigma_\ell.
\end{equation}
We say that a pair \((\boldsymbol\alpha, \boldsymbol\sigma)\) is weakly ordered if the indices \(\boldsymbol\alpha\) satisfy
\begin{equation}
\alpha_1 \le \alpha_2 \le \cdots \le \alpha_q.
\end{equation}
For every weakly–ordered pair $(\boldsymbol\alpha,\boldsymbol\sigma)$ put
\begin{equation}
\begin{aligned}
      \hat Q_{\boldsymbol\alpha,\boldsymbol\sigma}
      :=&|v_1\rangle \langle v_1|  \otimes  \cdots  \otimes  
        |v_q\rangle \langle v_q|,\\
      \quad 
      &n_j:=\bigl|\{\ell:\alpha_\ell=j\}\bigr|.
\end{aligned}
\end{equation}
Because the operator $\hat O^{\otimes q}$ is totally symmetric, the following weighted sum reproduces it exactly:
\begin{equation}
\label{eq:O-decomposition-correct}
   \hat O^{\otimes q}
    = 
   \lambda^{q}
   \sum_{\substack{\text{weakly ordered}\\(\boldsymbol\alpha,\boldsymbol\sigma)}}
      \frac{q!}{\displaystyle\prod_j n_j!} 
      \sgn(\boldsymbol\sigma) 
      \hat Q_{\boldsymbol\alpha,\boldsymbol\sigma},
\end{equation}

To clarify the notation used in the proof, let's consider a simple example. Let our system be a single qubit ($n=1$, so the Hilbert space dimension is $N=2^1=2$) and let the operator be $\hat{O} = \hat{\sigma}_z$. The eigenvalues are $\lambda = \pm 1$. The corresponding eigenspaces, indexed by $j \in \{1, \dots, N/2=1\}$, are spanned by:
\begin{itemize}
    \item Positive eigenspace ($j=1$): $\ket{p_1} = \ket{0}$ (for eigenvalue $+1$)
    \item Negative eigenspace ($j=1$): $\ket{m_1} = \ket{1}$ (for eigenvalue $-1$)
\end{itemize}
We will examine the second moment, which corresponds to setting $t=1$ and thus $q=2t=2$. The multi-index for the eigenspace, $\boldsymbol{\alpha} = (\alpha_1, \alpha_2)$, must be $(1,1)$ since $j$ can only be $1$. The sign multi-index $\boldsymbol{\sigma}=(\sigma_1, \sigma_2)$ can vary, leading to different operators $\hat{Q}_{\boldsymbol{\alpha},\boldsymbol{\sigma}}$.
\begin{itemize}
    \item \textbf{Case 1:} $\boldsymbol{\sigma} = (+, +)$
    \begin{itemize}
        \item The state vectors are $\ket{v_1} = \ket{p_1} = \ket{0}$ and $\ket{v_2} = \ket{p_1} = \ket{0}$.
        \item The operator is $\hat{Q}_{(1,1), (+,+)} = \ket{0}\bra{0} \otimes \ket{0}\bra{0} = \ket{00}\bra{00}$.
        \item The sign contribution is $\text{sgn}(\boldsymbol{\sigma}) = (+1)(+1) = +1$.
    \end{itemize}
    \item \textbf{Case 2:} $\boldsymbol{\sigma} = (+, -)$
    \begin{itemize}
        \item The state vectors are $\ket{v_1} = \ket{p_1} = \ket{0}$ and $\ket{v_2} = \ket{m_1} = \ket{1}$.
        \item The operator is $\hat{Q}_{(1,1), (+,-)} = \ket{0}\bra{0} \otimes \ket{1}\bra{1} = \ket{01}\bra{01}$.
        \item The sign contribution is $\text{sgn}(\boldsymbol{\sigma}) = (+1)(-1) = -1$.
    \end{itemize}
\end{itemize}
The other two cases, $\boldsymbol{\sigma} = (-,+)$ and $\boldsymbol{\sigma} = (-,-)$, are analogous. The full expansion of $(\hat{\sigma}_z)^{\otimes 2}$ is given by:
\begin{equation}
\begin{aligned}
    (\hat{\sigma}_z)^{\otimes 2} =& \left(\ket{0}\bra{0} - \ket{1}\bra{1}\right) \otimes \left(\ket{0}\bra{0} - \ket{1}\bra{1}\right) \\
    =& \ket{00}\bra{00} - \ket{01}\bra{01} - \ket{10}\bra{10} + \ket{11}\bra{11}
\end{aligned}
\end{equation}

Let \(|\psi\rangle\) be a Haar-random state on \(\mathbb C^{N}\).  Since  
\begin{equation}
   \mathbb E_{\psi}\bigl[|\psi\rangle   \langle\psi|^{\otimes q}\bigr]
   =\frac{\Pi_{\mathrm{sym}}
     }{\binom{N+q-1}{q}},
\end{equation}
with \(\Pi_{\mathrm{sym}}\) the projector onto the totally symmetric
subspace \(\vee^{q}   \mathbb C^{N}\), it follows that from
\eqref{eq:O-decomposition-correct}
\begin{equation}
\begin{aligned}
   \mathbb E_{\psi}   
   \bigl[\langle\psi|\hat O|\psi\rangle^{q}\bigr]
   =&
   \frac{\lambda^{q}}{\binom{N+q-1}{q}}
   \sum_{\substack{\text{weakly ordered}\\(\boldsymbol\alpha,\boldsymbol\sigma)}}
        \frac{q!}{\prod_{j}n_{j}!}  
        \sgn(\boldsymbol\sigma) 
        \operatorname{Tr}\bigl(
            \Pi_{\mathrm{sym}}    
            \hat Q_{\boldsymbol\alpha,\boldsymbol\sigma}\bigr)
\end{aligned}
\end{equation}
The overlap with \(\Pi_{\mathrm{sym}}\) depends only on how many signs in \(\boldsymbol\sigma\) are negative; a
short combinatorial count yields
\begin{equation}
\label{eq:Haar-moment-final}
   \mathbb E_{\psi}   
   \bigl[\langle\psi|\hat O|\psi\rangle^{q}\bigr]
     =  
   \lambda^{q}    
   \frac{\binom{N/2+t-1}{t}}{\binom{N+q-1}{q}} .
\end{equation}
For every non‑identity Pauli
\(\hat P\in\mathcal P_{n,k}\setminus\{\openone\}\) we have
\(\lambda=1\).
Replacing \(|\psi\rangle\) by
\(
   |\psi(\boldsymbol\vartheta)\rangle
   :=\hat U(\boldsymbol\vartheta)|\psi_{\mathrm{in}}\rangle
\)
and averaging over the unitary \(q\)-design
\(\{\hat U(\boldsymbol\vartheta)\}_{\boldsymbol\vartheta\in\Theta}\)
gives, by \eqref{eq:Haar-moment-final},
\begin{equation}
   \int_{\Theta}
      \bigl\langle\psi(\boldsymbol\vartheta)\bigl|\hat P
         \bigr|\psi(\boldsymbol\vartheta)\bigr\rangle^{q}
       \mathrm d\boldsymbol\vartheta
   =\frac{\binom{N/2+t-1}{t}}{\binom{N+q-1}{q}} .
\end{equation}
For the identity operator the same integral equals \(1\).
Splitting the sum over the low‑weight Pauli set
\(\mathcal P_{n,k}\) into identity plus
\(\operatorname{card}(\mathcal{P}_{n,k})-1\) traceless strings we arrive at
\begin{equation}
\begin{aligned}
    \int_{\Theta} \sum_{\hat{P} \in \mathcal{P}_{n,k}} \frac{|\langle\psi(\boldsymbol{\vartheta})|\hat{P}| \psi(\boldsymbol{\vartheta})\rangle|^q}{\operatorname{card}(\mathcal{P}_{n,k})} d\boldsymbol{\vartheta}&=\frac{1}{\operatorname{card}(\mathcal{P}_{n,k})} + \frac{(\operatorname{card}(\mathcal{P}_{n,k})-1)}{\operatorname{card}(\mathcal{P}_{n,k})}\frac{\binom{2^{n-1} +q/2 - 1}{q/2}}{\binom{2^{n} +q - 1}{q}} \\&= \frac{\binom{2^{n} +q - 1}{q} - \binom{2^{n-1} +q/2 - 1}{q/2}}{\operatorname{card}(\mathcal{P}_{n,k})\binom{2^{n} +q - 1}{q}}+\frac{\binom{2^{n-1} +q/2 - 1}{q/2}}{\binom{2^{n} +q - 1}{q}}.
\end{aligned}
\end{equation}
\end{proof}

\section{Dynamics of Deep Circuits}\label{Sec: Power Spectrum}

\begin{figure*}[t]
    \centering
    \includegraphics[width=16.3cm]{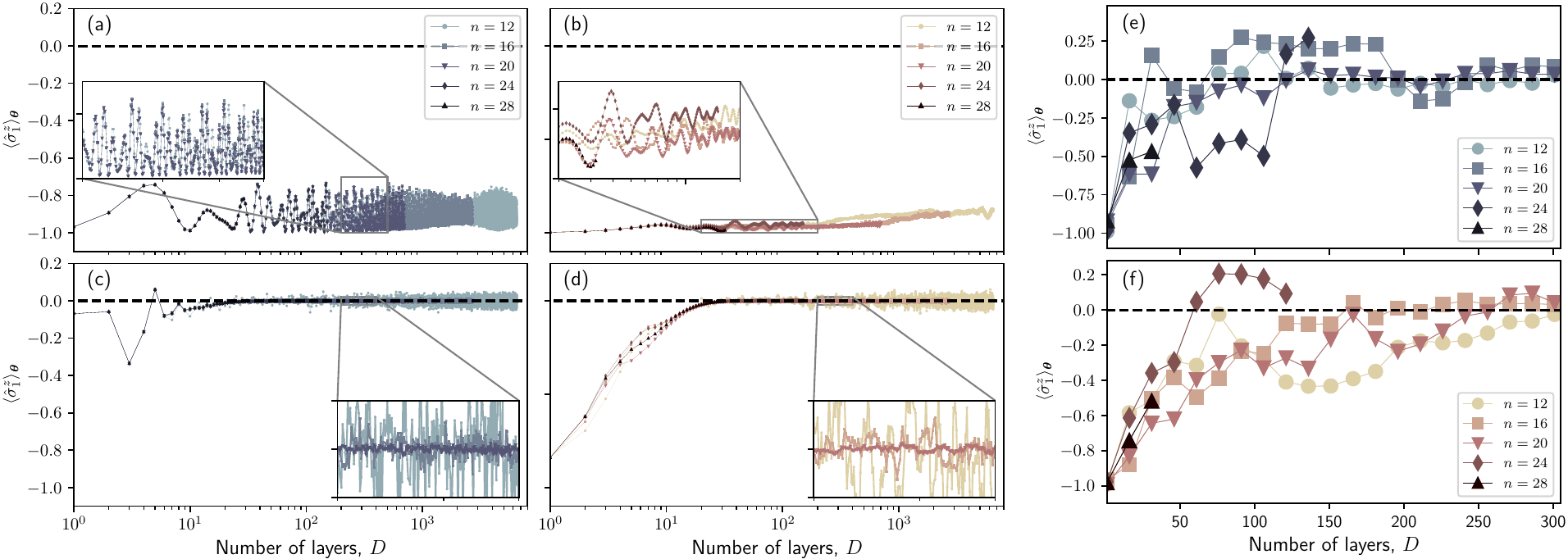}
    \caption{Demonstration of many-body localization within a circuit model, showcasing the output expectation values from deep variational quantum circuits across different phases and initializations. \textbf{(a, c)} Illustrate results for a 1D ring topology under MBL conditions with kick strengths $W = 1/5$ and $W = 7/5$ respectively, highlighting the preservation of local information even at greater circuit depths. \textbf{(b, d)} Display outcomes for a $\mathrm{Ci}_{n}(1,2)$ lattice topology with kick strengths $W = 1/10$ and $W = 7/10$ respectively, reflecting the dynamics across both MBL and thermal phases. \textbf{(e, f)} Compare these with results from non-Floquet-initialized circuits for the same topologies but with each parameter independently sampled from narrow ranges $[-1/5, 1/5]$ and $[-1/10, 1/10]$, respectively. These panels demonstrate a rapid loss of local information within about 100 layers, starkly contrasting with the robustness observed in the Floquet-initialized systems depicted in panels (a-d).}
    \label{fig:deep circuit}
\end{figure*}

In this section, we employ our Floquet initialization strategy while substantially extending the circuit depth to investigate the dynamics of deep quantum circuits. We examine the output expectation values of Pauli $\hat{\sigma}^z_j$ on the first qubit over a range of circuit depths $D$. We utilize both many-body localization (MBL) and thermal initializations to explore their distinct dynamics. Specifically, we apply a kick strength of $W = 0.2$ for the 1D ring topology and $W = 0.1$ for the $\mathrm{Ci}_{n}(1,2)$ lattice topology under MBL conditions, and set the strengths to $W = 1.4$ for the 1D ring and $W = 0.7$ for the $\mathrm{Ci}_{n}(1,2)$ under thermal conditions. The results are illustrated in Fig.~\ref{fig:deep circuit}(a-d). Remarkably, the local information in the MBL phase remains robustly preserved even beyond 1000 layers, underscoring the effectiveness of MBL dynamics in protecting initial system information within extremely deep circuits. This preservation sharply contrasts with the outcomes from circuits without Floquet initialization, where local information typically diminishes as circuit depth increases. We provide a comparative analysis by showing results from circuits with non-Floquet initialization, where each parameter is independently sampled from a narrow range of $[-0.2, 0.2]$ for the 1D ring and $[-0.1, 0.1]$ for the $\mathrm{Ci}_{n}(1,2)$ lattice. The results, depicted in Fig.~\ref{fig:deep circuit}(e,f), demonstrate that the initial local information is entirely lost after approximately 100 iterations. Despite significant fluctuations observed over extended periods, the MBL phase distinctly demonstrates its capacity to mitigate the loss of local information, reflecting the characteristic resilience of standard many-body localized systems.

Without loss of generality, we focus on the fluctuations of the output expectation values for circuit depths \(D\) ranging from 500 to 1500, using \(N=1000\) data points in a 16-qubit system. Initially, we apply linear regression to detrend the original data:
\begin{equation}
    y_D = \langle\hat{\sigma}_1^z\rangle_D - \overline{\langle \hat{\sigma}_1^z\rangle} - f(D),
\end{equation}
where \(f(D)\) represents the linear regression function fitted to the data. To mitigate spectral leakage in the power spectral density (PSD) analysis~\cite{Youngworth2005An}, we employ a window function:
\begin{equation}
    \tilde{y}_D = y_D \cdot w_D,
\end{equation}
where \(w_D\) is the Hanning window, defined as:
\begin{equation}
    w_D = \frac{1}{2} \left(1 - \cos\left(\frac{2 \pi D}{N - 1}\right)\right).
\end{equation}

Subsequently, we calculate the discrete Fourier transform of the windowed data using the fast Fourier transform algorithm:
\begin{equation}
      z_k = \sum_{D=0}^{N-1} \tilde{y}_D \cdot e^{-i \frac{2\pi k D}{N}},
\end{equation}
where \(  z_k\) represents the \(k\)-th element of the Fourier transform. The PSD, \(\textsc{psd}_k\), is determined by taking the square of the absolute value of \(  z_k\) and normalizing by the number of data points:
\begin{equation}
    \textsc{psd}_k = \frac{|  z_k|^2}{N}.
\end{equation}
The PSD is then normalized by dividing each element by the total power, which is the sum of all PSD values:
\begin{equation}
    \widetilde{\textsc{psd}}_k = \frac{\textsc{psd}_k}{\sum_{j=0}^{N/2} \textsc{psd}_j}.
\end{equation}
~\\
The resulting normalized spectrum, \(\widetilde{\textsc{psd}}_k\), illustrated in Fig.~\ref{fig: Power Spectrum}, reveals distinct characteristics of the circuit dynamics across different phases. In the MBL phase, the power spectrum displays several prominent peaks, indicative of multiple dominant frequency components. This suggests that the circuit's dynamical evolution is characterized by various periodic patterns with distinct frequencies. Conversely, in the thermal phase, the power spectrum exhibits a more uniform distribution of frequencies with fewer dominant components. This notable difference in spectral characteristics between the MBL and thermal phases provides a compelling diagnostic tool for characterizing and distinguishing the dynamical behavior of variational quantum circuits in different operational phases.

\begin{figure*}[t]
    \centering
    \includegraphics[width=13cm]{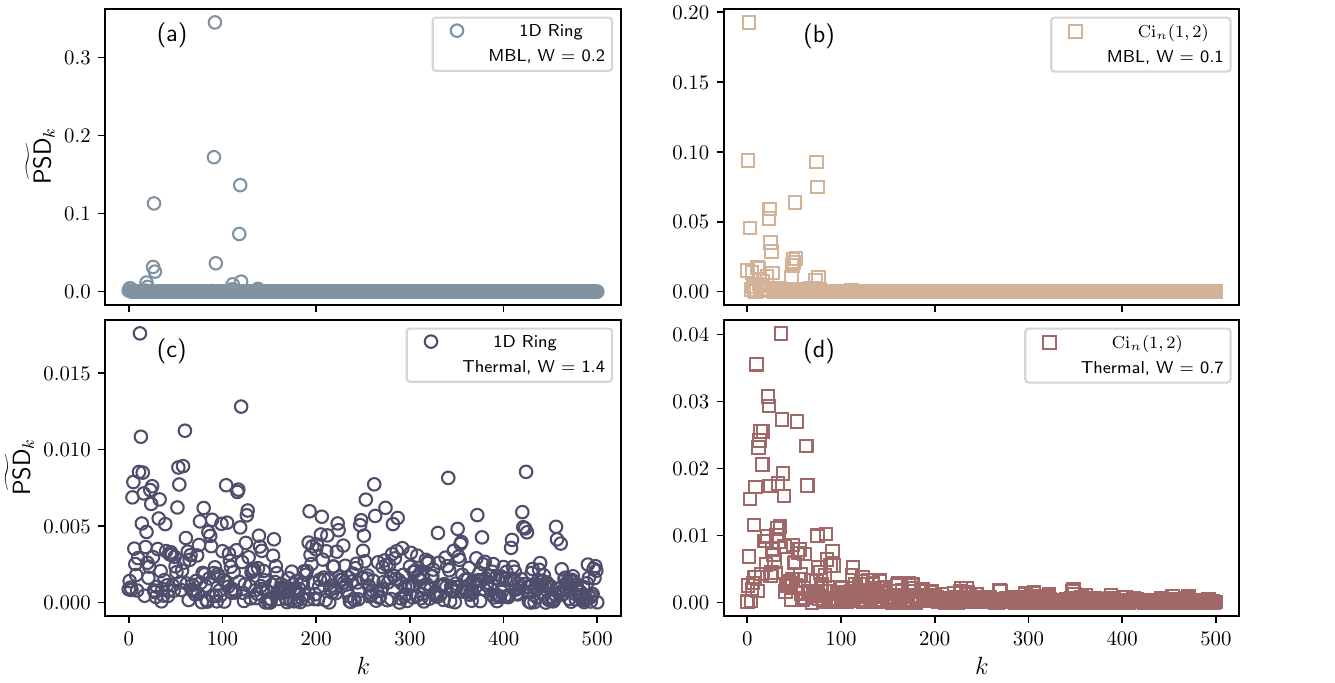}
    \caption{Power spectrum analysis of the output state expectation value fluctuations for $\hat{\sigma}^z_1$ in a 16-qubit system, with circuit depth $D$ varying from 100 to 500. This analysis continues using our Floquet initialization strategy, extending the circuit depth significantly. Panels \textbf{(a, c)} illustrate results for a 1D ring topology with kick strengths $W = 0.2, 1.4$, respectively; panels \textbf{(b, d)} display the $\mathrm{Ci}_{n}(1,2)$ lattice topology with kick strengths $W = 0.1, 0.7$, respectively.}
    \label{fig: Power Spectrum}
\end{figure*}

\section{Connecting the Gradient's \texorpdfstring{$\ell_\infty$}{l-infinity}-Norm to Component Variance}
\label{Appendix:inftynorm_variance}
\begin{proposition}\label{prop:inftynorm_variance}
Let $\boldsymbol{\vartheta}\in\mathbb{R}^{D(m_s + m_k)}$ be drawn from a distribution $\mathcal{D}$, and denote the gradient of the expectation value of the loss Hamiltonian by
\begin{equation}
\nabla  \langle\hat{\mathcal{H}}\rangle_{\boldsymbol\vartheta}
= \bigl(
\partial_{\vartheta_1}\langle\hat{\mathcal{H}}\rangle_{\boldsymbol\vartheta},
\dots,
\partial_{\vartheta_{D(m_s + m_k)}}\langle\hat{\mathcal{H}}\rangle_{\boldsymbol\vartheta}
\bigr).
\end{equation}
Assume the gradient vanishes on average, a standard premise in the analysis of barren plateaus:
\begin{equation}
\mathbb{E}_{\boldsymbol\vartheta\sim\mathcal{D}} \bigl[\nabla  \langle\hat{\mathcal{H}}\rangle_{\boldsymbol\vartheta}\bigr] = \boldsymbol{0}.
\end{equation}
If there exist constants $\alpha, p > 0$ such that
\begin{equation}
\Pr_{\boldsymbol\vartheta\sim\mathcal{D}} \left[
  \bigl\|\nabla  \langle\hat{\mathcal{H}}\rangle_{\boldsymbol\vartheta}\bigr\|_\infty \ge \alpha
\right] \ge p,
\end{equation}
then there exists at least one component $j \in \{1,\dots,D(m_s + m_k)\}$ such that
\begin{equation}
\operatorname{Var}_{\boldsymbol\vartheta\sim\mathcal{D}} 
\left[\partial_{\vartheta_j}\langle\hat{\mathcal{H}}\rangle_{\boldsymbol\vartheta}\right]
\ge \frac{p\alpha^2}{D(m_s + m_k)}.
\end{equation}
In particular, if both $p$ and $\alpha$ scale inverse-polynomially with the number of qubits~$n$, the gradient variance cannot decay faster than polynomially in~$n$, ruling out exponential barren-plateau behavior.
\end{proposition}

\begin{proof}
Since $\mathbb{E}_{\mathcal{D}}[\nabla  \langle\hat{\mathcal{H}}\rangle_{\boldsymbol\vartheta}] = \boldsymbol{0}$, we have
\begin{equation}
\operatorname{Var}_{\mathcal{D}}\left[\partial_{\vartheta_j}\langle\hat{\mathcal{H}}\rangle_{\boldsymbol\vartheta}\right]
= \mathbb{E}_{\mathcal{D}}\left[\left(\partial_{\vartheta_j}\langle\hat{\mathcal{H}}\rangle_{\boldsymbol\vartheta}\right)^2\right].
\end{equation}
Whenever $\|\nabla  \langle\hat{\mathcal{H}}\rangle_{\boldsymbol\vartheta}\|_\infty \ge \alpha$, there exists some $k$ such that $(\partial_{\vartheta_k}\langle\hat{\mathcal{H}}\rangle_{\boldsymbol\vartheta})^2 \ge \alpha^2$. This implies
\begin{equation}
\sum_{i=1}^{D(m_s + m_k)} (\partial_{\vartheta_i}\langle\hat{\mathcal{H}}\rangle_{\boldsymbol\vartheta})^2 \ge \alpha^2.
\end{equation}
Taking expectations,
\begin{equation}
\begin{aligned}
\sum_{i=1}^{D(m_s + m_k)} \operatorname{Var}_{\mathcal{D}}\left[\partial_{\vartheta_i}\langle\hat{\mathcal{H}}\rangle_{\boldsymbol\vartheta}\right] 
= \mathbb{E}_{\mathcal{D}}\left[\|\nabla  \langle\hat{\mathcal{H}}\rangle_{\boldsymbol\vartheta}\|_2^2\right]
\ge p \alpha^2,
\end{aligned}
\end{equation}
so at least one term in the sum must satisfy the lower bound claimed.
\end{proof}

For the Floquet-initialized parameter ensemble used in this work, kick angles are drawn from intervals symmetric about zero, and steady angles uniformly from $[-\pi, \pi)$. Combined with the parity properties of $\langle\hat{\mathcal{H}}\rangle_{\boldsymbol\vartheta}$, this ensures that $\mathbb{E}_{\mathcal{D}}[\nabla  \langle\hat{\mathcal{H}}\rangle_{\boldsymbol\vartheta}] = \boldsymbol{0}$, satisfying the premise of Proposition~\ref{prop:inftynorm_variance}. Therefore, the observed inverse-polynomial lower bound on $\|\nabla  \langle\hat{\mathcal{H}}\rangle_{\boldsymbol\vartheta}\|_\infty$ (see Fig.~\ref{fig:BP transition}) directly implies an inverse-polynomial lower bound on the gradient variance, reconciling our results with the standard variance-based barren plateau framework.

\section{Proof of Theorem~\ref{theorem:grad_bound}}
\label{appendix:grad_bound}

\begin{proof}
Let \(\boldsymbol{\delta} := \boldsymbol{\vartheta}_{\mathrm{th}} - \boldsymbol{\vartheta}_{\mathrm{MBL}}\), and define the linear path
\begin{equation}
\gamma(s) := (1 - s)\boldsymbol{\vartheta}_{\mathrm{MBL}} + s  \boldsymbol{\vartheta}_{\mathrm{th}}, \quad s \in [0,1].
\end{equation}
By the chain rule,
\begin{equation}
\label{eq:chain}
\frac{d}{ds} \langle \hat{\mathcal{H}} \rangle_{\gamma(s)} = \nabla_{\boldsymbol{\vartheta}} \langle \hat{\mathcal{H}} \rangle_{\gamma(s)} \cdot \boldsymbol{\delta}.
\end{equation}
Integrating both sides over \( s \in [0,1] \) and applying the triangle inequality yields
\begin{equation}
\label{eq:diff_bound}
\left| \langle \hat{\mathcal{H}} \rangle_{\boldsymbol{\vartheta}_{\mathrm{th}}} - \langle \hat{\mathcal{H}} \rangle_{\boldsymbol{\vartheta}_{\mathrm{MBL}}} \right|
\le \int_0^1 \left| \nabla_{\boldsymbol{\vartheta}} \langle \hat{\mathcal{H}} \rangle_{\gamma(s)} \cdot \boldsymbol{\delta} \right|    ds.
\end{equation}
Applying Hölder's inequality,
\begin{equation}
\left| \nabla_{\boldsymbol{\vartheta}} \langle \hat{\mathcal{H}} \rangle_{\gamma(s)} \cdot \boldsymbol{\delta} \right| \le \left\| \nabla_{\boldsymbol{\vartheta}} \langle \hat{\mathcal{H}} \rangle_{\gamma(s)} \right\|_\infty \cdot \|\boldsymbol{\delta}\|_1,
\end{equation}
we obtain
\begin{equation}
\label{eq:holder}
\left| \langle \hat{\mathcal{H}} \rangle_{\boldsymbol{\vartheta}_{\mathrm{th}}} - \langle \hat{\mathcal{H}} \rangle_{\boldsymbol{\vartheta}_{\mathrm{MBL}}} \right|
\le \|\boldsymbol{\delta}\|_1 \cdot \mathbb{E}_{s \in [0,1]}  \left\| \nabla_{\boldsymbol{\vartheta}} \langle \hat{\mathcal{H}} \rangle_{\gamma(s)} \right\|_\infty.
\end{equation}
Now we bound \( \|\boldsymbol{\delta}\|_1 \). By uniform kick scaling, each of the \( D m_k \) kick parameters is multiplied by a factor \( \lambda > 1 \), and satisfies
\begin{equation}
|\delta^{(\ell,j)}| = (\lambda - 1) |\vartheta^{(1,j)}_{\mathrm{MBL}}| \le (\lambda - 1) W,
\end{equation}
since the kick parameters in \( \boldsymbol{\vartheta}_{\mathrm{MBL}} \) lie in \([ -W, W ]\). Therefore,
\begin{equation}
\label{eq:delta_norm}
\|\boldsymbol{\delta}\|_1 \le D m_k (\lambda - 1) W.
\end{equation}
Substituting \eqref{eq:delta_norm} into \eqref{eq:holder}, and dividing both sides by \( D m_k (\lambda - 1) W \), we obtain
\begin{equation}
\mathbb{E}_{s \in [0,1]} \left\| \nabla_{\boldsymbol{\vartheta}} \langle \hat{\mathcal{H}} \rangle_{\gamma(s)} \right\|_\infty
 \gtrsim 
\frac{ \left| \langle \hat{\mathcal{H}} \rangle_{\boldsymbol{\vartheta}_{\mathrm{MBL}}} - \langle \hat{\mathcal{H}} \rangle_{\boldsymbol{\vartheta}_{\mathrm{th}}} \right| }
{ D m_k (\lambda - 1) W }.
\end{equation}
Under the assumption that \( \langle \hat{\mathcal{H}} \rangle_{\boldsymbol{\vartheta}_{\mathrm{MBL}}} = \Theta(1) \) and \( \langle \hat{\mathcal{H}} \rangle_{\boldsymbol{\vartheta}_{\mathrm{th}}} \) is exponentially small in \(n\), the numerator is \(\Theta(1)\), while the denominator grows at most polynomially in \(n\), since \( D, m_k = \operatorname{poly}(n) \). This establishes that the average gradient norm remains inverse-polynomial in \(n\), completing the proof.
\end{proof}

\section{Additional Benchmarks and Robustness Tests}\label{appendix:benchmark}
To evaluate the robustness of our approach, we analyze a more complex system characterized by the following Hamiltonian:
\begin{equation}
    \hat{\mathcal{H}} = \sum_{i<j} \frac{J \bigl[1 + \gamma(-1)^{i+j}\bigr]}{2|i-j|^{\alpha}}
\bigl(\hat \sigma_i^x\hat \sigma_j^x + \hat \sigma_i^y\hat \sigma_j^y
+ \Delta\hat \sigma_i^z\hat \sigma_j^z \bigr)
+ \sum_{i} h_i\hat \sigma_i^z.
\end{equation}
This Hamiltonian describes a one-dimensional spin-1/2 chain with long-range power-law-decaying interactions (controlled by the exponent $\alpha$), staggered couplings modulated by $\gamma$ (where $0 \leq \gamma \leq 1$; $\gamma=1$ implies couplings only at even distances, while $\gamma=0$ removes staggering), anisotropy parameter $\Delta$, and on-site magnetic fields $h_i$. For small values of $\alpha$ (e.g., $\alpha < 1.5$), the long-range interactions induce extended correlations that challenge classical tensor-network methods like density matrix renormalization group (DMRG)~\cite{White1992Density, Schollwock2011The}, as they lead to rapid entanglement growth and require exponentially large bond dimensions for accurate approximations. In contrast, larger $\alpha$ (e.g., $\alpha \geq 2$) effectively reduces to short-range behavior, where DMRG performs efficiently.

\begin{figure*}[b]
    \centering
    \includegraphics[width=16.5cm]{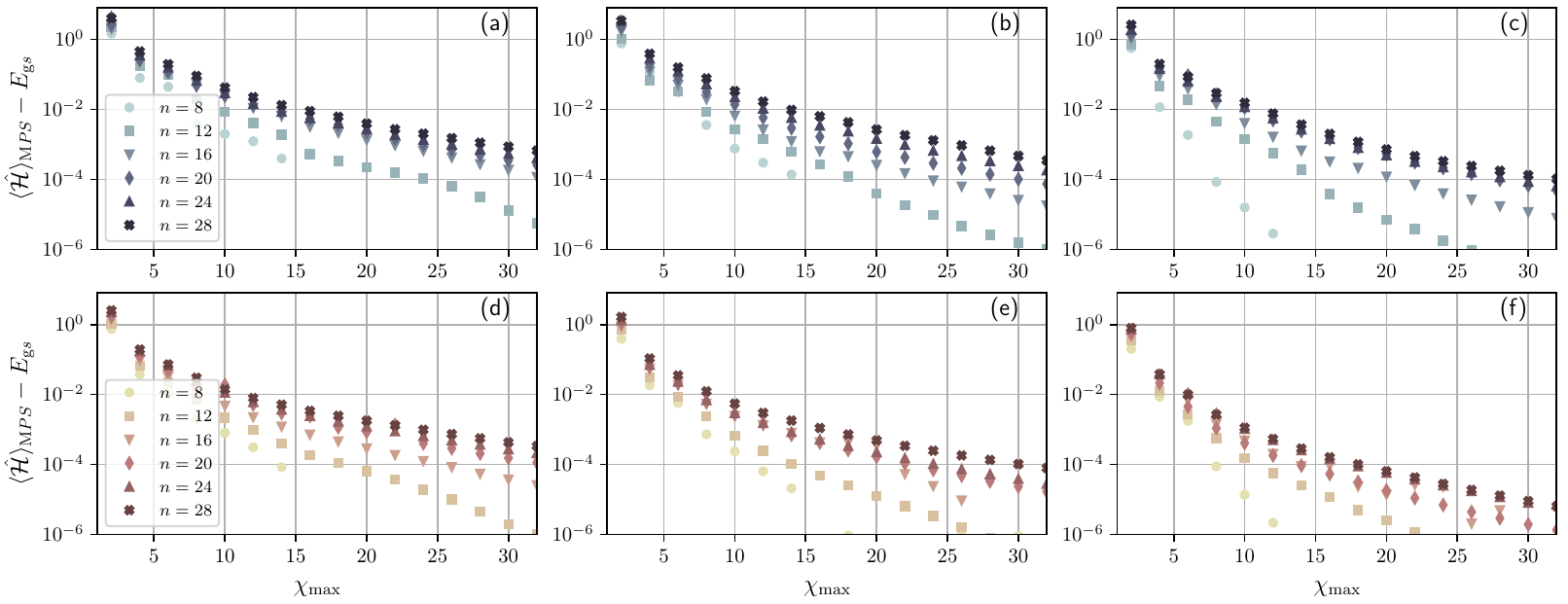}
    \caption{The absolute energy error relative to the exact ground state of MPS approximations for the long-range spin-chain Hamiltonian as a function of maximum bond dimension $\chi$. Panels (a)-(c) show results for fixed $\gamma = 3/4$, $\Delta = 2$, alternating magnetic fields $h_i = (-1)^i / 2$, while (d)-(f) correspond to sampled $\gamma \sim \text{Unif}(0.5,1.0)$, $\Delta \sim \text{Unif}(1.0,3.0)$, random fields $h_i$ uniformly sampled from $[-1,1]$ (averaged over 10 independent disorder realizations). Subpanels are organized by interaction decay exponent: (a,d) $\alpha=1.1$ (strong long-range), (b,e) $\alpha=1.5$ (intermediate), and (c,f) $\alpha=2.0$ (effectively short-range). Each curve represents a different system size $n=8,12,16,20,24,28$, with darker shades indicating larger $n$.}
    \label{fig:MPS-bonddim}
\end{figure*}

To assess the limitations of classical MPS approximations, we first optimize MPS trial states with varying maximum bond dimensions $\chi$, comparing the achieved energies to the exact ground-state energy. We fix $J = 1.0$. The results are shown in Fig.~\ref{fig:MPS-bonddim}, where panels (a)-(c) correspond to fixed $\gamma = 3/4$, $\Delta = 2$, alternating fields $h_i = (-1)^i / 2$, and panels (d)-(f) to sampled $\gamma \sim \text{Unif}(0.5,1.0)$, $\Delta \sim \text{Unif}(1.0,3.0)$, random fields with each $h_i$ independently sampled uniformly from $[-1, 1]$ (averaged over 10 independent disorder realizations). Panels (a,d) use $\alpha=1.1$, (b,e) $\alpha=1.5$, and (c,f) $\alpha=2.0$. In each subplot, the horizontal axis denotes the maximum bond dimension $\chi$, and the vertical axis shows the absolute energy error $|\langle \hat{\mathcal{H}} \rangle_{\text{MPS}} - E_{\text{gs}}|$, for system sizes $n=8,12,16,20,24,28$. As expected, larger system sizes require higher $\chi$ for achieving a target accuracy, with the effect being more pronounced at smaller $\alpha$ (e.g., $\alpha=1.1$), where long-range interactions degrade MPS performance.

Next, we compare the performance of our MBL-based initialization protocol against close-to-identity initialization, whose philosophy has been applied to achieve larger gradients~\cite{Zhang2022Escaping, Wang2024Trainability, Park2024Hamiltonian}, a common warm-start strategy in variational quantum algorithms that initializes circuit parameters near zero to preserve large gradients and proximity to the trial state. To evaluate their robustness against local minima, we employ the same hardware-efficient square circuit architecture (depth $D=n$) used throughout this work, along with an optimized MPS trial state with bond dimension $\chi=2$. For close-to-identity initialization, all circuit parameters across all layers are independently and uniformly sampled from $[-0.001, 0.001]$. For MBL initialization, we set $W=0.2$. Optimization proceeds via gradient descent until the energy fails to decrease by $10^{-4}$ over 50 consecutive iterations or reaches 1000 total iterations, at which point we record the final energy error relative to the exact ground state.

\begin{figure*}[t]
    \centering
    \includegraphics[width=13.0cm]{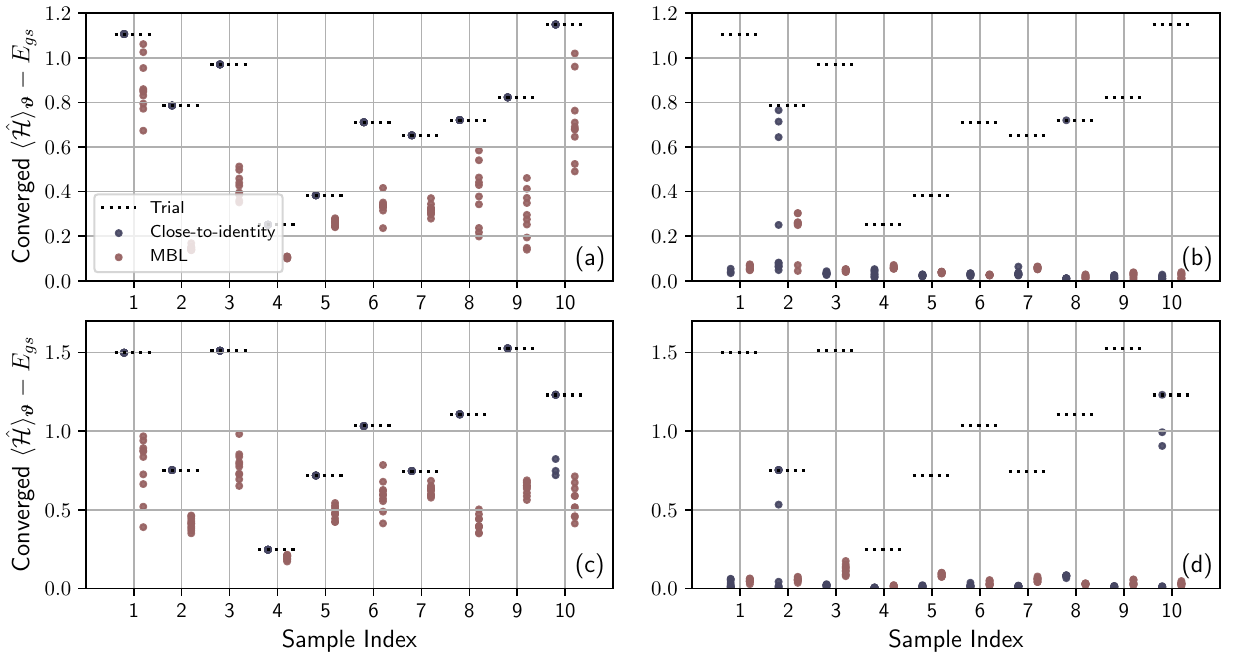}
    \caption{Comparison of converged energy errors for MBL initialization ($W=0.2$) versus close-to-identity initialization (parameters sampled from $[-0.001, 0.001]$) across 10 random Hamiltonian instances. Each bar represents the absolute energy error $|\langle \hat{\mathcal{H}} \rangle - E_{\text{gs}}|$ for a single instance, with dashed lines indicating the trial-state energy error. Panels (a) and (b) show results for $n=12$, while (c) and (d) are for $n=16$. Standard gradient descent is used in (a) and (c), and the Adam optimizer in (b) and (d). Close-to-identity initialization (blue circles) often remains trapped near the trial energy, whereas MBL initialization (red circles) more reliably reaches lower energies.}
    \label{fig:cii}
\end{figure*}

The results are presented in Fig.~\ref{fig:cii} for 10 independent random instances of the target Hamiltonian per subplot, with parameters sampled as follows: $\alpha \sim \text{Unif}(1,2)$, $\gamma \sim \text{Unif}(0.5,1.0)$, $\Delta \sim \text{Unif}(1.0,3.0)$, $J=1.0$, and on-site fields $h_i \sim \text{Unif}(-1,1)$. Panels (a) and (c) show results for $n=12$ and $n=16$ using standard gradient descent, while (b) and (d) employ the Adam optimizer~\cite{Kingma2017Adam} with hyperparameters \(\beta_1=0.9\), \(\beta_2=0.999\) and a fixed learning rate of 0.005. In panels (a) and (c), close-to-identity initialization frequently converges to the trial-state energy, indicating entrapment in local minima near the initial point. In contrast, MBL initialization often escapes these minima, achieving substantially lower energies due to its sampling of a broader parameter space while remaining in a low-entanglement regime. Panels (b) and (d) demonstrate improved performance for both methods with Adam, which better navigates the optimization landscape; however, several close-to-identity instances still fail to escape local minima, whereas MBL consistently reaches lower energies for these samples. These findings do not imply an unconditional advantage for MBL initialization but highlight scenarios where its structured exploration of parameter space facilitates escape from trial-state-induced local minima, enhancing robustness for challenging long-range models.

To investigate the sensitivity of our MBL-initialized protocol to the quality of the trial state and the impact of finite-shot noise in energy estimation, we conduct additional numerical experiments on the long-range spin-chain Hamiltonian. We consider three types of bond-dimension-2 MPS trial states: (i) an optimized MPS obtained via DMRG, (ii) a perturbed version of this optimized MPS, where its tensor parameters are interpolated towards a random MPS with a small interpolation factor $\gamma_{\text{interp}} = 0.05$ (representing mild errors in trial-state preparation), and (iii) a fully random MPS ($\gamma_{\text{interp}} = 1$), which serves as a poor-quality baseline. The interpolation is performed as $\theta_{\text{interp}} = (1 - \gamma_{\text{interp}}) \theta_{\text{opt}} + \gamma_{\text{interp}} \theta_{\text{rand}}$, where $\theta_{\text{opt}}$ and $\theta_{\text{rand}}$ are the tensor parameters of the optimized and random MPS, respectively. For each trial state, we apply the MBL-initialized variational circuit (depth $D=n=12$, $W=0.2$) and optimize using gradient descent with a step size of 0.005.

Energy expectations during optimization are estimated using finite shots, incorporating Pauli grouping~\cite{Kandala2017Hardware, Crawford2021Efficient} to reduce measurement overhead by simultaneously measuring commuting Pauli terms. Shot allocation is optimized adaptively based on the variance of each grouped observable. Specifically, for a Hamiltonian decomposed as $\hat{\mathcal{H}} = \sum_k c_k \hat{P}_k$ with $\hat{P}_k$ Pauli strings grouped into $G$ commuting sets, the total variance of the estimator is bounded by $\operatorname{Var}(\langle \hat{\mathcal H} \rangle) \leq \sum_{g=1}^G (\sum_{k \in g} |c_k| \sqrt{\operatorname{Var}(\hat{P}_k)})^2 / S_g$, where $S_g$ is the number of shots allocated to group $g$. We use 2048 shots per evaluation, distributed proportionally to the group variances to minimize the overall estimator variance, with the allocation recomputed every 10 iterations.

\begin{figure*}[t]
    \centering
    \includegraphics[width=16.5cm]{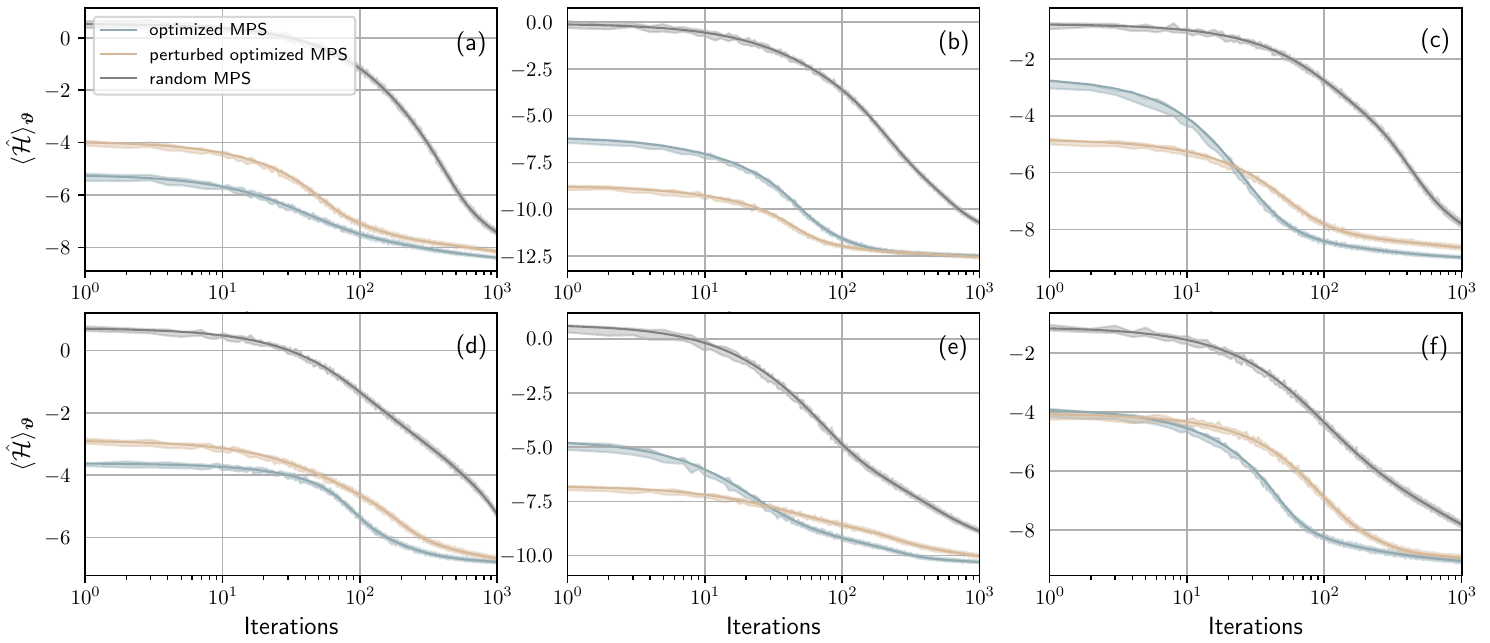}
    \caption{Optimization trajectories for the long-range spin-chain Hamiltonian ($n=12$) under MBL initialization, showing sensitivity to trial-state quality and finite-shot effects across six independent Hamiltonian instances (panels a-f). Each panel displays three curves: optimized bond-2 MPS trial (blue), perturbed MPS ($\gamma_{\text{interp}}=0.05$; brown), and random MPS ($\gamma_{\text{interp}}=1$; gray). Solid lines denote exact (infinite-shot) energies, while shaded bands represent the mean $\pm$ standard deviation from 10 finite-shot repetitions (2048 shots per evaluation, with Pauli grouping and variance-based allocation).}
    \label{fig:perturb}
\end{figure*}

The optimization trajectories are shown in Fig.~\ref{fig:perturb} for six independent Hamiltonian instances (sampled as $\alpha \sim \text{Unif}(1,2)$, $\gamma \sim \text{Unif}(0.5,1.0)$, $\Delta \sim \text{Unif}(1.0,3.0)$, $J=1.0$, $h_i \sim \text{Unif}(-1,1)$), each in a separate panel (a)-(f). In each subplot, the three curves correspond to the optimized MPS trial (blue), perturbed MPS (orange), and random MPS (gray). Solid lines represent exact (infinite-shot) energies, while shaded bands indicate the mean $\pm$ one standard deviation from 10 repeated finite-shot evaluations at each step. Mild perturbations ($\gamma_{\text{interp}}=0.05$) have minimal impact on convergence, with final energies closely matching those from the optimized trial. In contrast, fully random trials always lead to slower convergence and higher final errors, underscoring the importance of a reasonable initial approximation. Finite-shot noise introduces fluctuations, but with only 2048 shots, the estimates remain sufficiently accurate for guiding optimization, as the shot-induced variance scales as $\mathcal{O}(\|\hat{\mathcal{H}}\|^2 / S)$ (where $S$ is the total shots) and is independent of system size for Hamiltonians with bounded-norm terms. This independence arises because the estimator variance for a $k$-local Hamiltonian is bounded by $\mathcal{O}(\|\hat{\mathcal{H}}\|^2 / S)$, where $\|\hat{\mathcal{H}}\|$ grows polynomially with $n$ but can be mitigated through grouping and allocation without exponential overhead.

These results demonstrate that the MBL protocol exhibits robustness to moderate errors in the trial state, allowing for imperfect classical approximations without significant degradation. Moreover, the low shot requirements confirm practical viability on noisy intermediate-scale quantum devices, where measurement budgets are limited.

\end{appendices}

\end{document}